\documentclass[a4paper,12pt]{article}

\usepackage[T1]{fontenc}
\usepackage[utf8]{inputenc}
\usepackage[english]{babel}
\usepackage{amsmath}
\usepackage{amssymb}
\usepackage{amsthm}
\usepackage{tikz}
\usepackage{placeins}
\usetikzlibrary{arrows.meta}
\usepackage[hidelinks]{hyperref}
\hypersetup{
  pdftitle={Integral Representations for Solutions of the Linearized Fourth Painleve Equation},
  pdfauthor={Oleg M. Kiselev},
  pdfsubject={Linearized Painleve IV equation and monodromy data},
  pdfkeywords={Painleve IV, linearized equation, Lax pair, squared eigenfunctions, monodromy data, Stokes multipliers}
}
\newtheorem{proposition}{Proposition}
\newtheorem{lemma}{Lemma}
\DeclareMathOperator{\diag}{diag}

\title{Integral Representations for Solutions of the\\
Linearized Fourth Painlev\'e Equation}
\author{Oleg M. Kiselev\\
\small Innopolis University, Innopolis, Russia\\
\small \texttt{o.kiselev@innopolis.ru}}
\date{}

\begin{document}

\maketitle

\begin{abstract}
We formulate the linearization of the fourth Painlev\'e equation as a
linearized Hamiltonian system and transform the scalar variational equation
to normal form, with the first-derivative term eliminated. From the canonical
fundamental solutions of the Jimbo--Miwa Lax pair, we construct quadratic
expressions $Q_{IV}$ and $R_{IV}$ that satisfy the identity
$\partial_x^2Q_{IV}-U_{IV}Q_{IV}=\partial_\lambda R_{IV}$.
Here $U_{IV}$ is the coefficient in the normal-form equation. A rapid-decay
cycle with a vanishing boundary term yields an integral solution of the
linearized equation. Every nontrivial such cycle has asymptotic branches in at
least two distinct Stokes sectors; otherwise the integral vanishes by Cauchy's
theorem. A second rapid-decay homology class gives another integral solution.
The Wronskian of the two integral solutions depends holomorphically on the
monodromy data, so a single nonzero value implies generic linear
independence. For one set of complex data, direct quadrature gives numerical
evidence of a nonzero Wronskian. We also derive
integral formulas for variations of the Stokes multipliers, the connection
matrix, and the local monodromy, and verify the contour representation by
substitution into the linearized equation.
\end{abstract}

\noindent\textbf{Keywords:} Painlev\'e IV equation; linearized equation;
Lax pair; squared eigenfunctions; rapid-decay cycles; monodromy data;
Stokes multipliers.

\section{Introduction}

Linearized equations describe the dependence of a nonlinear solution on
initial data, parameters, and perturbations. In singularly perturbed problems,
their solutions determine variations of amplitudes, phases, and other
asymptotic parameters. Explicit representations of linearized solutions are
therefore needed to derive connection formulas relating asymptotic regimes
before and after a dynamic bifurcation.

The Painlev\'e equations provide a natural class of models for this problem.
They arise both as reductions of integrable evolution equations and as
universal local models of critical transitions. An isomonodromic
representation relates a solution to its monodromy data and relates
variations of these data to solutions of the linearized equation.

An early asymptotic formulation of this problem is due to Haberman. He
identified the second Painlev\'e transcendent in the inner layer of a nonlinear
transition problem \cite{Haberman1977}. When a parameter passes slowly through
a critical value of an algebraic bifurcation, the local scale leads, depending
on the type of degeneracy, to the first or second Painlev\'e transcendent
\cite{Haberman1979}. Connection formulas between the regimes before and after
the bifurcation are thereby reduced to the analysis of distinguished
solutions of Painlev\'e equations.

The program of constructing nonlinear analogues of the special functions of
wave catastrophes was formulated in early work by Suleimanov
\cite{Suleimanov1992,SuleimanovHabibullin1993,Suleimanov1994}. Canonical
oscillatory integrals in linear theory describe universal profiles near
caustics. In nonlinear integrable problems, this role is played by special
solutions associated with isomonodromic systems and Painlev\'e equations. In
\cite{KiselevSuleimanov1999}, the asymptotics of solutions of all six
Painlev\'e equations were organized into a hierarchy of Hamiltonian systems.
Painlev\'e transcendents as nonlinear special functions were subsequently
discussed systematically in \cite{Clarkson2003}, and the program of nonlinear
analogues of wave-catastrophe functions was developed further in
\cite{Suleimanov2019}. Its extension to variational equations requires compatible explicit
representations for solutions of the linearized equations and for variations
of the monodromy data.

Such formulas are needed in dynamic-bifurcation and autoresonance problems.
For hard loss of stability in $P_{II}$, an algebraic asymptotic regime is
matched to a rapidly oscillatory one through an inner layer governed by $P_I$
\cite{KiselevHardLoss2001}. During slow passage through the nonhyperbolic
homoclinic orbit associated with a subcritical pitchfork bifurcation, one
inner regime is governed by $P_{II}$ \cite{Haberman2001}; $P_I$ plays an
analogous role for the unfolding of a saddle--center bifurcation
\cite{DiminnieHaberman2002}. In slow separatrix crossing and capture into
autoresonance, inner scalings lead to Painlev\'e equations and determine
connection formulas between asymptotic regimes
\cite{KiselevGlebov2003,KiselevGlebov2007,GlebovKiselevTarkhanov2010}.
Capture into resonance, separatrix scattering, the measure of captured
trajectories, phase-locking conditions, and the dependence of a bifurcation
boundary on a perturbation are expressed through parameters of asymptotic
families and their connection formulas
\cite{KiselevTarkhanov2014Chaos,KiselevTarkhanov2014JMP,
Kiselev2018ND,Kiselev2019RND,Kiselev2021Chaos}.

The derivative of a Painlev\'e transcendent with respect to a monodromy
parameter satisfies the linearized equation, while variations of the
monodromy data are expressed by integrals of quadratic combinations of
Lax-pair solutions. This leads to integral representations in terms of
squared-eigenfunction kernels. The monodromy formulation and the Riemann--Hilbert
method are presented in
\cite{ItsNovokshenov1986,FokasItsKapaevNovokshenov2006}.

For the second Painlev\'e equation, this scheme was implemented in
\cite{Kiselev2024RCD}. The corresponding formula is
\[
    v(x)=\int_\Gamma
    \bigl(\Psi_{11}^2+\Psi_{21}^2\bigr)(\lambda,x)\,d\lambda,
    \qquad
    \Gamma:\ \infty_6\longrightarrow\infty_1.
\]
Here $\Psi$ is a canonical fundamental solution of the $P_{II}$ Lax pair, while
$\infty_6$ and $\infty_1$ denote the lifted asymptotic directions of the two
contour branches lying in distinct Stokes sectors. If both tails lie in one
connected component of the decay region and the contour can be closed by an
arc at infinity within a domain where the integrand is analytic, Cauchy's
theorem makes the integral vanish. Separate deformations of the tails make
the boundary values of the total $\lambda$-derivative vanish.

\paragraph{Problem statement.}
The purpose of this paper is to construct integral solutions of the
variational equation along a fixed $P_{IV}$ solution and to express
variations of the monodromy data. We take the canonical fundamental
solutions of the associated linear system in $\lambda$ as given functions
and use their quadratic combinations as integral kernels.
The spectral system has an irregular singularity of Poincar\'e rank two at
$\lambda=\infty$ and a Fuchsian singularity at $\lambda=0$.
We use Volterra representations of these canonical solutions with one
asymptotic integration limit at infinity. Solutions of the linearized
equation are represented by integrals over rapid-decay chains on
the universal cover of
$\mathbb C\setminus\{0\}$.

\paragraph{Main result.}
Let $(q(x),p(x))$ be a holomorphic solution of the Hamiltonian $P_{IV}$
system on a simply connected domain, with $q(x)\ne0$ throughout that domain.
Fix a branch of $\sqrt q$ there, set $h=2p-x-q/2$, and take any column
$\psi=(\psi_1,\psi_2)^T$ of a solution of the Jimbo--Miwa Lax pair. Then
\[
    \partial_x^2Q_{IV}=U_{IV}(x)Q_{IV}+\partial_\lambda R_{IV},
\]
where $\alpha$ and $\beta$ are the parameters of the scalar $P_{IV}$ equation,
and
\[
\begin{aligned}
    Q_{IV}
    &=\frac{\sqrt q}{2\lambda}
      \bigl((h-\lambda)\psi_1^2-\psi_1\psi_2\bigr),\\
    R_{IV}
    &=-\sqrt q\left(\lambda+2x+\frac{3q}{2}\right)\psi_1^2,\\
    U_{IV}
    &=x^2-\alpha+6xq+\frac{15}{4}q^2
      -\frac{3\beta}{2q^2}.
\end{aligned}
\]
Let $\Gamma$ be the rapid-decay cycle on the universal cover of
$\mathbb C\setminus\{0\}$ defined in Section~6. Its homology class is fixed
as $x$ varies, its support meets at least two components of the decay region,
and
\[
    R_{IV}\big|_{\partial\Gamma}=0.
\]
Then
\[
    y(x)=\int_\Gamma Q_{IV}(\lambda,x)\,d\lambda
\]
solves the linearized equation in normal form
(\ref{eqLinP4SelfAdjointArticle}), and
$v(x)=\sqrt{q(x)}\,y(x)$ solves the original linearized equation
(\ref{eqLinP4ScalarArticle}). A different rapid-decay homology class gives a
second integral solution. Its Wronskian with the first solution depends holomorphically on
local coordinates on the monodromy manifold. If it is nonzero for at least
one set of data, the two integral solutions form a fundamental pair
generically. Section~6 gives numerical evidence for this condition for one
set of complex data. The nondegeneracy statement is conditional on an exact
nonzero value.

We also derive componentwise Volterra equations for canonical fundamental
solutions, exact expressions for the four Stokes multipliers through the
solutions of these equations, the global relation among
$s_1,\ldots,s_4$, $\theta_0$, and $\theta_\infty$, and integral formulas
for infinitesimal variations of the Stokes multipliers, the connection
matrix, and the local monodromy.

\paragraph{Organization of the paper.}
Section~2 fixes the Hamiltonian normalization and the parameters. Section~3
contains the scalar and Hamiltonian linearizations and the Liouville
transformation. In Section~4, the Lax pair is written explicitly and the
phase and first two coefficients of its formal solution are computed.
Section~5 defines the monodromy data, the Volterra equations, and integral
formulas for $s_1,\ldots,s_4$, including the geometry of progressive
contours. Section~6 proves the identity for $Q_{IV}$ and $R_{IV}$,
constructs integral solutions, states a criterion for their linear independence,
and implements the integral representation numerically. Section~7 derives
variations of the monodromy data and proves convergence of the corresponding
integrals. Section~8 relates the result to the Hamiltonian hierarchy of
asymptotic problems.

\section{Hamiltonian form of \texorpdfstring{$P_{IV}$}{PIV}}

The independent variable $x\in\mathbb C$ plays the role of time in the
Hamiltonian system and is the isomonodromic deformation parameter. We denote
the canonical coordinate and momentum by $q=q(x)$ and $p=p(x)$,
respectively. These functions are assumed to be meromorphic in $x$, and a
prime always denotes differentiation with respect to $x$.

The quantities $\theta_0,\theta_\infty\in\mathbb C$ are fixed parameters,
independent of both $x$ and the spectral variable $\lambda$. In the Lax pair
below, $\pm\theta_0$ are the local exponents at the Fuchsian singularity
$\lambda=0$, whereas $\mp\theta_\infty$ are the formal exponents at the
irregular singularity $\lambda=\infty$. The isomonodromic evolution in $x$
preserves these parameters.

We use the Hamiltonian normalization compatible with the Jimbo--Miwa Lax pair
\cite{JimboMiwa1981,Okamoto1986},
\begin{equation}
    q'=\frac{\partial H_{IV}}{\partial p},\qquad
    p'=-\frac{\partial H_{IV}}{\partial q},
    \label{eqP4HamiltonSystemArticle}
\end{equation}
where
\begin{equation}
    H_{IV}
    =
    2p^2q-\frac18q^3-\frac12xq^2
    +\frac12(2\theta_\infty-1-x^2)q
    -\frac{2\theta_0^2}{q}.
    \label{eqP4HamiltonianArticle}
\end{equation}
Thus,
\begin{equation}
    q'=4pq,\qquad
    p'=
    -2p^2+\frac38q^2+xq-\frac12(2\theta_\infty-1-x^2)
    -\frac{2\theta_0^2}{q^2}.
    \label{eqP4HamiltonEquationsArticle}
\end{equation}
Eliminating $p$ gives the fourth Painlev\'e equation
\begin{equation}
    q''
    =
    \frac{(q')^2}{2q}
    +\frac32 q^3
    +4xq^2
    +2(x^2-\alpha)q
    +\frac{\beta}{q},
    \label{eqP4ScalarArticle}
\end{equation}
where
\begin{equation}
    \alpha=2\theta_\infty-1,\qquad
    \beta=-8\theta_0^2 .
    \label{eqP4ParametersArticle}
\end{equation}
In the scalar equation, $\theta_\infty$ determines $\alpha$ and
$\theta_0^2$ determines $\beta$; the monodromy problem retains the local
exponent $\theta_0$ itself.

\section{Linearization}

Let $(q(x;\varepsilon),p(x;\varepsilon))$ be a one-parameter family of
solutions of the Hamiltonian system (\ref{eqP4HamiltonEquationsArticle}). Its
variation is the vector
\[
    \eta=
    \begin{pmatrix}\delta q\\ \delta p\end{pmatrix},
    \qquad
    \delta q=\left.\frac{\partial q}{\partial\varepsilon}\right|_{\varepsilon=0},
    \qquad
    \delta p=\left.\frac{\partial p}{\partial\varepsilon}\right|_{\varepsilon=0}.
\]
For the scalar equation, set $v=\delta q$, so that
$q(x;\varepsilon)=q(x;0)+\varepsilon v(x)+O(\varepsilon^2)$. All
calculations are local in a domain where $q(x;0)$ is finite and nonzero.

\begin{proposition}
The function $v$ satisfies the linearized equation
\begin{equation}
    v''
    =
    \frac{q'}{q}v'
    +
    \left(
        -\frac{(q')^2}{2q^2}
        +\frac92q^2
        +8xq
        +2(x^2-\alpha)
        -\frac{\beta}{q^2}
    \right)v .
    \label{eqLinP4ScalarArticle}
\end{equation}
\end{proposition}

In Hamiltonian variables, the linearization has the canonical form
\begin{equation}
    \frac{d}{dx}
    \begin{pmatrix}
        \delta q\\
        \delta p
    \end{pmatrix}
    =
    \begin{pmatrix}
        4p & 4q\\
        \frac34q+x+\frac{4\theta_0^2}{q^3} & -4p
    \end{pmatrix}
    \begin{pmatrix}
        \delta q\\
        \delta p
    \end{pmatrix}.
    \label{eqLinP4HamiltonianArticle}
\end{equation}
Introduce the symplectic matrix $J$ and the Hessian with respect to $(q,p)$:
\[
    J=
    \begin{pmatrix}0&1\\-1&0\end{pmatrix}.
\]
\[
    \operatorname{Hess}_{(q,p)}H_{IV}
    :=
    \begin{pmatrix}
        \dfrac{\partial^2H_{IV}}{\partial q^2}&
        \dfrac{\partial^2H_{IV}}{\partial q\partial p}\\[2mm]
        \dfrac{\partial^2H_{IV}}{\partial p\partial q}&
        \dfrac{\partial^2H_{IV}}{\partial p^2}
    \end{pmatrix}.
\]
For the Hamiltonian (\ref{eqP4HamiltonianArticle}), this matrix is
\[
    \operatorname{Hess}_{(q,p)}H_{IV}
    =
    \begin{pmatrix}
        -\dfrac34q-x-\dfrac{4\theta_0^2}{q^3}&4p\\[2mm]
        4p&4q
    \end{pmatrix}.
\]
The variable $x$ is held fixed when
$\operatorname{Hess}_{(q,p)}H_{IV}$ is evaluated. System
(\ref{eqLinP4HamiltonianArticle}) can be written as
\[
    \eta'=J\operatorname{Hess}_{(q,p)}H_{IV}\,\eta
\].
The quadratic Hamiltonian of the linearized Hamiltonian system is
\begin{equation}
    \begin{aligned}
        K_{IV}
        &=
        \frac12
        \begin{pmatrix}
            \delta q & \delta p
        \end{pmatrix}
        \operatorname{Hess}_{(q,p)}H_{IV}
        \begin{pmatrix}
            \delta q\\
            \delta p
        \end{pmatrix}
        \\[1mm]
        &=
        2q(\delta p)^2+4p\,\delta q\,\delta p
        -
        \left(
            \frac38q+\frac{x}{2}+\frac{2\theta_0^2}{q^3}
        \right)(\delta q)^2.
    \end{aligned}
    \label{eqP4TangentHamiltonianArticle}
\end{equation}
$K_{IV}=K_{IV}(\delta q,\delta p;x)$ is the Hamiltonian for the variation
$(\delta q,\delta p)$ along the fixed solution $(q(x),p(x))$, because
\[
    (\delta q)'=\frac{\partial K_{IV}}{\partial(\delta p)},
    \qquad
    (\delta p)'=-\frac{\partial K_{IV}}{\partial(\delta q)},
\]
coincide with (\ref{eqLinP4HamiltonianArticle}).

The scalar equation (\ref{eqLinP4ScalarArticle}) contains a first-derivative
term. It can be removed by the standard Liouville transformation. In a simply
connected domain where $q\ne0$, fix a branch of $\sqrt q$ and set
\begin{equation}
    v=\sqrt{q}\,y .
    \label{eqP4SelfAdjointSubstitutionArticle}
\end{equation}
Then
\[
    v''-\frac{q'}qv'
    =
    \sqrt q\left[
        y''+
        \left(
            \frac{q''}{2q}-\frac{3(q')^2}{4q^2}
        \right)y
    \right].
\]
Substitution of (\ref{eqLinP4ScalarArticle}) and elimination of $q''$ by
(\ref{eqP4ScalarArticle}) give the normal form, with the first-derivative
term eliminated:
\begin{equation}
    y''
    =
    \left(
        x^2-\alpha
        +6xq
        +\frac{15}{4}q^2
        -\frac{3\beta}{2q^2}
    \right)y .
    \label{eqLinP4SelfAdjointArticle}
\end{equation}
In terms of the parameters (\ref{eqP4ParametersArticle}), the last term is
$12\theta_0^2/q^2$.

Denoting the coefficient in parentheses by $U_{IV}(x)$, write
(\ref{eqLinP4SelfAdjointArticle}) in operator form as
\[
    \mathcal L_{IV}y=0,
    \qquad
    \mathcal L_{IV}:=-\frac{d^2}{dx^2}+U_{IV}(x).
\]
The differential expression $\mathcal L_{IV}$ is formally symmetric with
respect to the bilinear pairing, without complex conjugation: for any twice
differentiable functions $f$ and $g$, Lagrange's identity gives
\[
    f\mathcal L_{IV}g-(\mathcal L_{IV}f)g
    =
    \frac{d}{dx}\bigl(f'g-fg'\bigr).
\]
After integration, the difference between the two bilinear expressions is
determined by the boundary term $f'g-fg'$.

\section{The Lax pair and the monodromy problem}

Set $\sigma_3=\diag(1,-1)$. In the Jimbo--Miwa normalization, the spectral
part of the $P_{IV}$ Lax pair has the rational form
\begin{equation}
    \frac{\partial\Psi}{\partial\lambda}
    =
    A_{IV}(\lambda,x)\Psi,
    \qquad
    A_{IV}(\lambda,x)
    =
    \lambda\sigma_3+A_0(x)+\frac{A_{-1}(x)}{\lambda}.
    \label{eqP4LaxStructureArticle}
\end{equation}
The diagonal gauge factor $k(x)$ used in the general formulation is set equal
to one below
\cite{JimboMiwa1981,FokasItsKapaevNovokshenov2006,VanDerPutTop2013,Babich2026}.
In this gauge, set
\begin{equation}
    h=2p-x-\frac q2.
    \label{eqP4hArticle}
\end{equation}
Then
\begin{equation}
    A_{IV}(\lambda,x)
    =
    \begin{pmatrix}
        a(\lambda) & b(\lambda)\\
        c(\lambda) & -a(\lambda)
    \end{pmatrix},
    \label{eqP4AabcArticle}
\end{equation}
where
\begin{eqnarray}
    a(\lambda)&=&\lambda+x+\frac{qh}{2\lambda},
    \label{eqP4aArticle}
    \\
    b(\lambda)&=&1-\frac{q}{2\lambda},
    \label{eqP4bArticle}
    \\
    c(\lambda)&=&-qh-2\theta_\infty
    +\frac{qh^2-4\theta_0^2/q}{2\lambda}.
    \label{eqP4cArticle}
\end{eqnarray}
Equivalently, $A_0$ and $A_{-1}$ are
\begin{equation}
    A_0=
    \begin{pmatrix}
        x & 1\\
        -qh-2\theta_\infty & -x
    \end{pmatrix},
    \qquad
    A_{-1}
    =
    \frac12
    \begin{pmatrix}
        qh & -q\\
        qh^2-4\theta_0^2/q & -qh
    \end{pmatrix}.
    \label{eqP4A0AminusArticle}
\end{equation}
The eigenvalues of $A_{-1}$ are $\pm\theta_0$; hence $\lambda=0$ is a
Fuchsian singular point. At $\lambda=\infty$, the system has an irregular
singularity of Poincar\'e rank two: with $z=1/\lambda$, the coefficient
matrix is $-\sigma_3/z^3-A_0/z^2-A_{-1}/z$.

The deformation part of the Lax pair is
\[
    \frac{\partial\Psi}{\partial x}=B_{IV}(\lambda,x)\Psi,
    \qquad
    B_{IV}=
    \begin{pmatrix}
        \lambda+x+q/2&1\\
        -qh-2\theta_\infty&-\lambda-x-q/2
    \end{pmatrix}.
\]
The zero-curvature condition
\[
    \partial_xA_{IV}-\partial_\lambda B_{IV}
    +[A_{IV},B_{IV}]=0
\]
is equivalent to the system
\[
    q'=q(2h+q+2x),
    \qquad
    h'=-h^2-2h(q+x)-2\theta_\infty
       -\frac{4\theta_0^2}{q^2},
\]
which also follows from (\ref{eqP4HamiltonEquationsArticle}) and the
definition (\ref{eqP4hArticle}).

We now determine the phase of the formal solution at $\lambda=\infty$.
Since $A_{IV}$ is traceless, its eigenvalues are
$\pm\mu(\lambda,x)$, where the branch $\mu\sim\lambda$ is selected and
\[
    \mu^2=-\det A_{IV}=a^2+bc.
\]
Equations (\ref{eqP4aArticle})--(\ref{eqP4cArticle}) give directly
\[
    a^2
    =
    \lambda^2+2x\lambda+x^2+qh+O(\lambda^{-1}),
    \qquad
    bc=-qh-2\theta_\infty+O(\lambda^{-1}).
\]
Consequently,
\[
    \mu^2
    =
    \lambda^2+2x\lambda+x^2-2\theta_\infty+O(\lambda^{-1}),
\]
and expansion of the square root yields
\[
    \mu
    =
    \lambda+x-\frac{\theta_\infty}{\lambda}
    +O(\lambda^{-2}).
\]
The formal diagonalizing gauge has the form
\[
    T(\lambda,x)=I+O(\lambda^{-1}).
\]
Because $T^{-1}T_\lambda=O(\lambda^{-2})$, the terms of orders
$\lambda$, $1$, and $\lambda^{-1}$ obtained above integrate to the
exponential part of the formal solution:
\[
    \Psi_{\mathrm{form}}(\lambda,x)
    =
    \left(I+\frac{\Psi_1(x)}{\lambda}
    +\frac{\Psi_2(x)}{\lambda^2}+\cdots\right)
    \exp\left(\Theta_{IV}(\lambda,x)\sigma_3\right).
\]
The scalar phase $\Theta_{IV}$ is defined by
\[
    \frac{\partial\Theta_{IV}}{\partial\lambda}
    =
    \lambda+x-\frac{\theta_\infty}{\lambda}.
\]
On the selected branch of the logarithm, it is therefore
\begin{equation}
    \Theta_{IV}(\lambda,x)
    =
    \frac{\lambda^2}{2}+x\lambda-\theta_\infty\log\lambda .
    \label{eqP4PhaseArticle}
\end{equation}
We also compute the first matrix coefficients of the formal solution. Set
\[
    F(\lambda,x)
    =
    I+\sum_{n\geq1}\frac{F_n(x)}{\lambda^n},
    \qquad
    F_0=I,\qquad F_{-1}=0.
\]
Substitution of $\Psi_{\mathrm{form}}=F e^{\Theta_{IV}\sigma_3}$ into
(\ref{eqP4LaxStructureArticle}) and comparison of the coefficients of
$\lambda^{-n}$ give the recurrence
\[
\begin{split}
    [F_{n+1},\sigma_3]
    &+xF_n\sigma_3-A_0F_n
    -\theta_\infty F_{n-1}\sigma_3\\
    &-A_{-1}F_{n-1}-(n-1)F_{n-1}=0,
    \qquad n=0,1,\ldots .
\end{split}
\]
Here $F_n=\Psi_n$. The diagonal part of each $F_n$ is determined by the
solvability condition for the next recurrence equation. For brevity, define
\[
\begin{aligned}
    \chi&=qh+2\theta_\infty,\\
    \rho
    &=
    \frac{qh^2}{4}+\frac{q^2h}{4}+\frac{xqh}{2}
    +\frac{q\theta_\infty}{2}+x\theta_\infty
    -\frac{\theta_0^2}{q},\\
    \Delta
    &=x\rho+\frac{\theta_\infty^2-\theta_0^2}{2},
    \qquad
    \varkappa=\rho^2+\frac{\chi}{4}.
\end{aligned}
\]
Then the first two coefficients are
\[
    \Psi_1
    =
    \begin{pmatrix}
        -\rho&-\dfrac12\\[1mm]
        -\dfrac{\chi}{2}&\rho
    \end{pmatrix},
\]
\[
    \Psi_2
    =
    \begin{pmatrix}
        \dfrac{\varkappa+\Delta}{2}
        &
        \dfrac{x+q/2-\rho}{2}\\[2mm]
        \rho\left(1+\dfrac{\chi}{2}\right)-\dfrac{q\chi}{4}
        &
        \dfrac{\varkappa-\Delta}{2}
    \end{pmatrix}.
\]
The coefficient comparison is carried through order $\lambda^{-3}$; the
diagonal of $\Psi_2$ is fixed by the solvability condition for the next
recurrence equation. In particular,
\[
    \operatorname{tr}\Psi_1=0,\qquad
    \operatorname{tr}\Psi_2=\varkappa,\qquad
    \det\Psi_1=-\varkappa,
\]
so that
\[
    \det\left(
        I+\frac{\Psi_1}{\lambda}+\frac{\Psi_2}{\lambda^2}
    \right)
    =
    1+O(\lambda^{-3}).
\]
Equivalently, the exponential factor is
\[
    e^{(\lambda^2/2+x\lambda)\sigma_3}
    \lambda^{-\theta_\infty\sigma_3}.
\]
The negative powers of $\lambda$ not included in
(\ref{eqP4PhaseArticle}) belong to the matrix prefactor. This is the standard
formal normalization of the Jimbo--Miwa system
\cite{JimboMiwa1981,FokasItsKapaevNovokshenov2006}; the sign of
$\theta_\infty\log\lambda$ is also fixed unambiguously by
(\ref{eqP4aArticle})--(\ref{eqP4cArticle}).

The quadratic term $\lambda^2/2$ produces four Stokes sectors of opening
$\pi/2$. Under a positive circuit about the origin,
$\log\lambda\mapsto\log\lambda+2\pi i$; hence the logarithmic term in
(\ref{eqP4PhaseArticle}) gives the formal monodromy factor
$e^{-2\pi i\theta_\infty\sigma_3}$. Together with the Fuchsian point
$\lambda=0$, these data determine the geometry of the contours used in the
integral formulas.

\paragraph{Spectral and deformation normalizations.}
The matrices $\Psi_k=F_k e^{\Theta_{IV}\sigma_3}$ below have spectral
normalization $F_k\to I$. Put $b_0=x+q/2$. Zero curvature and comparison
of the canonical asymptotics give
\[
 \partial_x\Psi_k=B_{IV}\Psi_k-\Psi_k b_0\sigma_3.
\]
Indeed, $\Psi_k^{-1}(B_{IV}\Psi_k-\partial_x\Psi_k)$ is independent of
$\lambda$; its value is fixed by the sectorial normalization.
On a simply connected pole-free $x$-domain define
\[
 \mathfrak d(x)=\exp\left(\int_{x_0}^x(t+q(t)/2)\,dt\right),\quad
 \mathsf D=\diag(\mathfrak d,\mathfrak d^{-1}),\quad \Phi_k=\Psi_k\mathsf D.
\]
Then $\Phi_{k,\lambda}=A_{IV}\Phi_k$ and $\Phi_{k,x}=B_{IV}\Phi_k$.
The contour integrals in Section~6 use columns of $\Phi_k$.
For $\Psi_{k+1}=\Psi_kS_k$, the constant Stokes matrices are
\[
 \widetilde S_k=\mathsf D^{-1}S_k\mathsf D,\qquad
 S_k'=[b_0\sigma_3,S_k],\qquad \widetilde S_k'=0.
\]
Thus $s_k'=-2b_0s_k$ for $k=1,3$ and $s_k'=2b_0s_k$ for $k=2,4$.
Sections~5 and~7 use spectral normalization at fixed $x$. At $x=x_0$,
$\mathsf D=I$ and $S_k=\widetilde S_k$.
The monodromy coordinates in Section~6 refer to this reference
normalization.

\section{Monodromy data}

Denote the canonical fundamental solutions near infinity by
$\Psi_k$, $k=1,\ldots,4$. They are normalized by the asymptotic condition
\begin{equation}
    \Psi_k(\lambda)
    \sim
    \left(I+O(\lambda^{-1})\right)
    \exp\left(\Theta_{IV}(\lambda,x)\sigma_3\right),
    \qquad
    \lambda\to\infty,\quad \lambda\in\Omega_k,
    \label{eqP4CanonicalInfinityArticle}
\end{equation}
where the canonical domains $\Omega_k$ are fixed by the phase-plane
construction below and have asymptotic opening $\pi/2$. With the chosen
normalization, the Stokes matrices are defined by
\begin{equation}
    \Psi_{k+1}=\Psi_k S_k,\qquad k=1,\ldots,4,
    \label{eqP4StokesConnectionArticle}
\end{equation}
where the solutions are considered on the universal cover of a punctured
neighborhood of infinity. The matrices $S_k$ are triangular:
\begin{equation}
    S_1=\begin{pmatrix}1&0\\ s_1&1\end{pmatrix},\quad
    S_2=\begin{pmatrix}1&s_2\\0&1\end{pmatrix},\quad
    S_3=\begin{pmatrix}1&0\\ s_3&1\end{pmatrix},\quad
    S_4=\begin{pmatrix}1&s_4\\0&1\end{pmatrix}.
    \label{eqP4StokesMatricesArticle}
\end{equation}
Set
\begin{equation}
    F_\infty=e^{-2\pi i\theta_\infty\sigma_3}
    =
    \begin{pmatrix}
        \alpha_\infty&0\\0&\alpha_\infty^{-1}
    \end{pmatrix},
    \qquad
    \alpha_\infty=e^{-2\pi i\theta_\infty}.
    \label{eqP4FormalMonodromyInfinityArticle}
\end{equation}
The sectors are numbered clockwise, as specified below. Set
$P=S_1S_2S_3S_4$. Continuing $\Psi_1$ counterclockwise once about zero
gives $\Psi_1(\lambda e^{2\pi i})=\Psi_5(\lambda)F_\infty$.
Thus, in the basis $\Psi_1$, the counterclockwise monodromy about zero
and the clockwise monodromy about infinity are
\begin{equation}
    M_0=P F_\infty,\qquad
    M_\infty=F_\infty^{-1}P^{-1}=M_0^{-1}.
    \label{eqP4MonodromyInfinityArticle}
\end{equation}

For the local connection formulas assume $2\theta_0\notin\mathbb Z$.
Resonant local systems require a Levelt normalization and are outside the
scope of those formulas. Near the origin, choose a local solution
\begin{equation}
    \Psi^{(0)}(\lambda)
    =
    G_0(\lambda)\lambda^{\theta_0\sigma_3},
    \qquad
    G_0(0)\in SL(2,\mathbb C),
    \label{eqP4LocalZeroArticle}
\end{equation}
and connection matrices $C_k$:
\begin{equation}
    \Psi_k(\lambda)=\Psi^{(0)}(\lambda)C_k .
    \label{eqP4ConnectionMatrixArticle}
\end{equation}
In the same basis, the local monodromy at $\lambda=0$ is
\begin{equation}
    M_0=C_1^{-1}e^{2\pi i\theta_0\sigma_3}C_1 .
    \label{eqP4MonodromyZeroArticle}
\end{equation}
With the adopted orientation of the circuits, the global monodromy relation
on the sphere is $M_0M_\infty=I$. In particular, the trace of $M_\infty$ is
fixed by the local exponent $\theta_0$:
\begin{equation}
    \operatorname{tr}M_\infty=2\cos(2\pi\theta_0).
    \label{eqP4TraceRelationArticle}
\end{equation}
Since a unimodular $2\times2$ matrix and its inverse have the same trace,
equation (\ref{eqP4MonodromyInfinityArticle}) gives the relation
\begin{equation}
    \alpha_\infty^2s_2s_3+\alpha_\infty^2
    +s_1s_2s_3s_4+s_1s_2+s_1s_4+s_3s_4+1
    =
    2\alpha_\infty\cos(2\pi\theta_0).
    \label{eqP4StokesRelationArticle}
\end{equation}
The monodromy data are specified by
\[
    \theta_0,\quad\theta_\infty,\quad s_1,s_2,s_3,s_4,\quad C_k,
\]
subject to (\ref{eqP4StokesRelationArticle}) and diagonal gauge equivalence.
The matrix $C_k$ relates a canonical solution normalized at
$\lambda=\infty$ to a local Fuchsian solution at $\lambda=0$.

\subsection*{Integral construction of the Stokes matrices}

We define the canonical solutions by equivalent Volterra integral equations
\cite{Wasow1987,FokasItsKapaevNovokshenov2006}; the multipliers
$s_1,\ldots,s_4$ then follow from
(\ref{eqP4StokesConnectionArticle}).

Take the finite approximation to the formal prefactor obtained above,
\[
    \widehat F(\lambda,x)
    =
    I+\frac{\Psi_1(x)}{\lambda}
    +\frac{\Psi_2(x)}{\lambda^2}
\]
and write out all its entries:
\[
    \widehat F=(f_{ij})_{i,j=1}^2,
    \qquad
    \delta_F=\det\widehat F=f_{11}f_{22}-f_{12}f_{21},
\]
\[
\begin{aligned}
    f_{11}
    &=1-\frac{\rho}{\lambda}
      +\frac{\varkappa+\Delta}{2\lambda^2},&
    f_{12}
    &=-\frac{1}{2\lambda}
      +\frac{x+q/2-\rho}{2\lambda^2},\\
    f_{21}
    &=-\frac{\chi}{2\lambda}
      +\frac{\rho(1+\chi/2)-q\chi/4}{\lambda^2},&
    f_{22}
    &=1+\frac{\rho}{\lambda}
      +\frac{\varkappa-\Delta}{2\lambda^2}.
\end{aligned}
\]
Here $q=q(x)$ and $h=h(x)$, while $\chi$, $\rho$, $\Delta$, and
$\varkappa$ were defined above; they are held constant under differentiation
with respect to $\lambda$. Hence,
\[
\begin{aligned}
    \partial_\lambda f_{11}
    &=\frac{\rho}{\lambda^2}
      -\frac{\varkappa+\Delta}{\lambda^3},&
    \partial_\lambda f_{12}
    &=\frac{1}{2\lambda^2}
      -\frac{x+q/2-\rho}{\lambda^3},\\
    \partial_\lambda f_{21}
    &=\frac{\chi}{2\lambda^2}
      -\frac{2\rho(1+\chi/2)-q\chi/2}{\lambda^3},&
    \partial_\lambda f_{22}
    &=-\frac{\rho}{\lambda^2}
      -\frac{\varkappa-\Delta}{\lambda^3}.
\end{aligned}
\]
Set also
\[
    E(\lambda,x)
    =
    e^{\Theta_{IV}(\lambda,x)\sigma_3},
    \qquad
    D(\lambda,x)
    =
    \frac{\partial\Theta_{IV}}{\partial\lambda}\sigma_3.
\]
Define the residual matrix
\[
    \mathcal R(\lambda,x)
    =
    \widehat F^{-1}
    \left(A_{IV}\widehat F-\frac{\partial\widehat F}{\partial\lambda}\right)
    -D.
\]
Componentwise, set
\[
\begin{aligned}
    Q_{11}&=af_{11}+bf_{21}-\partial_\lambda f_{11},&
    Q_{12}&=af_{12}+bf_{22}-\partial_\lambda f_{12},\\
    Q_{21}&=cf_{11}-af_{21}-\partial_\lambda f_{21},&
    Q_{22}&=cf_{12}-af_{22}-\partial_\lambda f_{22},
\end{aligned}
\]
where $a$, $b$, and $c$ are given by
(\ref{eqP4aArticle})--(\ref{eqP4cArticle}). Then
\[
    \mathcal R=
    \begin{pmatrix}
        \mathcal R_{11}&\mathcal R_{12}\\
        \mathcal R_{21}&\mathcal R_{22}
    \end{pmatrix}
\]
consists of the explicitly specified scalar functions
\[
\begin{aligned}
    \mathcal R_{11}
    &=\frac{f_{22}Q_{11}-f_{12}Q_{21}}{\delta_F}
      -\left(\lambda+x-\frac{\theta_\infty}{\lambda}\right),\\
    \mathcal R_{12}
    &=\frac{f_{22}Q_{12}-f_{12}Q_{22}}{\delta_F},\\
    \mathcal R_{21}
    &=\frac{-f_{21}Q_{11}+f_{11}Q_{21}}{\delta_F},\\
    \mathcal R_{22}
    &=\frac{-f_{21}Q_{12}+f_{11}Q_{22}}{\delta_F}
      +\left(\lambda+x-\frac{\theta_\infty}{\lambda}\right).
\end{aligned}
\]
The recurrence for $\Psi_1$ and $\Psi_2$ shows that
\[
    \mathcal R(\lambda,x)=O(\lambda^{-2}),
    \qquad \lambda\to\infty.
\]
For sufficiently large $|\lambda|$, the matrix $\widehat F$ is invertible.
After the substitution
\[
    \Psi=\widehat F\,Z\,E
\]
the spectral equation of the Lax pair is equivalent to
\[
    \frac{\partial Z}{\partial\lambda}
    =
    [D,Z]+\mathcal RZ.
\]

Set $\varepsilon_1=1$ and $\varepsilon_2=-1$. For each Stokes sector
$\Omega_k$, choose progressive contours $\Gamma_k^{ij}(\lambda)$ running
from the corresponding asymptotic direction at infinity to $\lambda$. On
matrix-valued functions, define the operator
\[
    (\mathcal K_k Z)_{ij}(\lambda)
    =
    \int_{\Gamma_k^{ij}(\lambda)}
    e^{(\varepsilon_i-\varepsilon_j)
    (\Theta_{IV}(\lambda)-\Theta_{IV}(\mu))}
    (\mathcal R(\mu)Z(\mu))_{ij}\,d\mu .
\]
The dependence on $x$ is suppressed in this and the following formulas. The
canonical solution in $\Omega_k$ is determined by the Volterra equation
\[
    Z_k=I+\mathcal K_kZ_k,
    \qquad
    \Psi_k=\widehat F\,Z_kE.
\]
If
\[
    Z_k=(z_{ij}^{(k)})_{i,j=1}^2,
\]
then this matrix equation is equivalent to the following four scalar integral
equations:
\[
\begin{aligned}
z_{11}^{(k)}(\lambda)
&=1+\int_{\Gamma_k^{11}(\lambda)}
  \bigl(\mathcal R_{11}z_{11}^{(k)}
       +\mathcal R_{12}z_{21}^{(k)}\bigr)(\mu)\,d\mu,\\
z_{12}^{(k)}(\lambda)
&=\int_{\Gamma_k^{12}(\lambda)}
  e^{2(\Theta_{IV}(\lambda)-\Theta_{IV}(\mu))}
  \bigl(\mathcal R_{11}z_{12}^{(k)}
       +\mathcal R_{12}z_{22}^{(k)}\bigr)(\mu)\,d\mu,\\
z_{21}^{(k)}(\lambda)
&=\int_{\Gamma_k^{21}(\lambda)}
  e^{-2(\Theta_{IV}(\lambda)-\Theta_{IV}(\mu))}
  \bigl(\mathcal R_{21}z_{11}^{(k)}
       +\mathcal R_{22}z_{21}^{(k)}\bigr)(\mu)\,d\mu,\\
z_{22}^{(k)}(\lambda)
&=1+\int_{\Gamma_k^{22}(\lambda)}
  \bigl(\mathcal R_{21}z_{12}^{(k)}
       +\mathcal R_{22}z_{22}^{(k)}\bigr)(\mu)\,d\mu.
\end{aligned}
\]
In the right-hand sides, every function under the integral sign is evaluated
at $\mu$, whereas the exterior factors $\Theta_{IV}(\lambda)$ refer to the
finite endpoint of the contour.

The contour $\Gamma_k^{ij}(\lambda)$ has one asymptotic branch and the finite
endpoint $\lambda$, and it may lie in a single Stokes sector. Oriented
differences of such paths have two asymptotic branches; a nonzero difference
connects distinct connected components of the decay region.

We now specify the contours precisely. Fix a small $\delta>0$ and a
sufficiently large $R$. The asymptotic relation
\[
    \operatorname{Re}\Theta_{IV}(\mu)
    =\frac{r^2}{2}\cos(2\varphi)+O(r),
    \qquad \mu=re^{i\varphi},\quad r\to\infty,
\]
shows that the boundaries of the exponential-decay regions are the four
Stokes rays
\[
    \arg\mu=\frac{\pi}{4}+\frac{m\pi}{2},
    \qquad m=0,1,2,3.
\]
The asymptotic tail of $\Gamma_k^{12}(\lambda)$ is chosen inside one of the
sectors
\[
    |\arg\mu-m\pi|<\frac{\pi}{4}-\delta,
    \qquad m\in\mathbb Z,
\]
where
\[
    \operatorname{Re}\Theta_{IV}(\mu)\longrightarrow+\infty,
\]
whereas the tail of $\Gamma_k^{21}(\lambda)$ is chosen inside one of the
sectors
\[
    \left|\arg\mu-\left(\frac{\pi}{2}+m\pi\right)\right|
    <\frac{\pi}{4}-\delta,
    \qquad m\in\mathbb Z,
\]
where
\[
    \operatorname{Re}\Theta_{IV}(\mu)\longrightarrow-\infty.
\]
The contour is oriented \emph{from infinity toward $\lambda$} and consists
of a radial tail
\[
    \mu=r e^{i\varphi},\qquad +\infty>r\geq R,
\]
with a fixed admissible $\varphi$, followed by a finite smooth arc from
$Re^{i\varphi}$ to $\lambda$. Along a progressive tail,
$\operatorname{Re}\Theta_{IV}(\mu)$ varies monotonically from the relevant
infinite limit to a finite value. Therefore, the exponential factors are
\[
    e^{2(\Theta_{IV}(\lambda)-\Theta_{IV}(\mu))}
    \quad\hbox{and}\quad
    e^{-2(\Theta_{IV}(\lambda)-\Theta_{IV}(\mu))}
\]
and decay on the asymptotic tails. For diagonal entries there is no
exponential factor; the contours $\Gamma_k^{11}$ and $\Gamma_k^{22}$ may be
taken to coincide with either admissible contour, and absolute convergence
follows from $\mathcal R=O(\mu^{-2})$. The finite part of the path does not
affect convergence and may be deformed, with fixed endpoints, within a domain
where the integrand is analytic.

If the tail for the $12$ component is chosen in the horizontal sector around
the negative real semiaxis, then, for the cut $(-\infty,0]$, one takes
$\varphi=\pi-\eta$ or $\varphi=-\pi+\eta$, with
$0<\eta<\pi/4-\delta$. Thus the path follows one bank of the cut but not the
cut itself. The two banks correspond to different values of $\log\mu$ and,
after gluing, belong to adjacent sheets of the universal cover.

On a Stokes ray, the leading exponential factor oscillates and does not
provide uniform decay. Its boundary value is defined as the limit
$\varphi\to(\pi/4+m\pi/2)\pm0$ from the chosen sector. The contours for two
adjacent canonical solutions pass on the corresponding sides of the boundary
ray. In the single-integral term of the Neumann series, the resulting
two-ended difference connects two distinct components of the decay region. A
difference whose ends lie in the same component can be closed by a vanishing
arc and is zero by Cauchy's theorem. This geometry is shown in
Figure~\ref{figP4ProgressiveContoursArticle}.

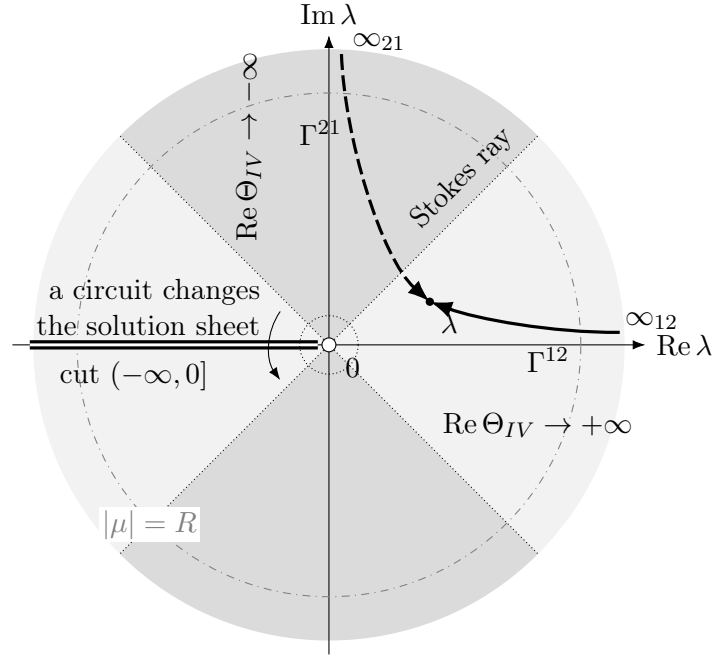
\begin{figure}[htbp]
\centering
\begin{tikzpicture}[scale=0.92,>=Latex,font=\small]
    \fill[black!5] (0,0)--(-45:4.25) arc (-45:45:4.25)--cycle;
    \fill[black!5] (0,0)--(135:4.25) arc (135:225:4.25)--cycle;
    \fill[black!14] (0,0)--(45:4.25) arc (45:135:4.25)--cycle;
    \fill[black!14] (0,0)--(225:4.25) arc (225:315:4.25)--cycle;

    \draw[->] (-4.55,0)--(4.55,0) node[right] {$\operatorname{Re}\lambda$};
    \draw[->] (0,-4.45)--(0,4.45) node[above] {$\operatorname{Im}\lambda$};
    \foreach \ang in {45,135,225,315}
        \draw[black,densely dotted] (0,0)--(\ang:4.25);
    \draw[gray,dash dot] (0,0) circle (3.62);
    \node[gray,fill=white,inner sep=1pt] at (-2.58,-2.58) {$|\mu|=R$};

    \draw[black,very thick,double=white,double distance=1.1pt]
        (-4.3,0)--(-0.16,0);
    \node[below] at (-2.8,-0.08) {cut $(-\infty,0]$};
    \draw[densely dotted] (0,0) circle (0.42);
    \fill[white] (0,0) circle (0.10);
    \draw (0,0) circle (0.10);
    \node[below right] at (0.08,-0.04) {$0$};

    \coordinate (P) at (1.45,0.62);
    \fill (P) circle (1.7pt) node[below right] {$\lambda$};
    \draw[black,very thick,->]
        (4.18,0.18) .. controls (3.2,0.17) and (2.1,0.38) .. (P);
    \node[right] at (4.10,0.35) {$\infty_{12}$};
    \node[below] at (3.15,0.12) {$\Gamma^{12}$};
    \draw[black,very thick,dash pattern=on 6pt off 2pt,->]
        (0.18,4.18) .. controls (0.22,3.0) and (0.75,1.25) .. (P);
    \node[above right] at (0.18,4.08) {$\infty_{21}$};
    \node[left] at (0.34,3.05) {$\Gamma^{21}$};

    \node at (3.0,-1.15)
        {$\operatorname{Re}\Theta_{IV}\to+\infty$};
    \node[rotate=90] at (-1.15,2.85)
        {$\operatorname{Re}\Theta_{IV}\to-\infty$};
    \node[rotate=45,above] at (45:3.05) {Stokes ray};
    \draw[->,thin] (-0.70,0.48) arc (145:215:0.86);
    \node[left,align=right] at (-0.83,0.55)
        {a circuit changes\\the solution sheet};
\end{tikzpicture}
\caption{Geometry of the progressive contours. The light-gray regions have
$\operatorname{Re}\Theta_{IV}\to+\infty$ and contain admissible tails for
the $12$ component; the dark-gray regions have
$\operatorname{Re}\Theta_{IV}\to-\infty$ and contain tails for the $21$
component. Dotted lines denote Stokes rays, the double line is the selected
cut, and the solid and dashed curves are the contours for the $12$ and $21$
components.}
\label{figP4ProgressiveContoursArticle}
\end{figure}

The matrix $A_{IV}(\lambda)$ is single-valued on the punctured plane and has
a Fuchsian singularity at the origin. For the fixed branch
$-\pi<\arg\lambda<\pi$, multivaluedness of the solutions is described by
\[
    E=e^{(\lambda^2/2+x\lambda)\sigma_3}
      \lambda^{-\theta_\infty\sigma_3},
    \qquad
    \Psi^{(0)}=G_0(\lambda)\lambda^{\theta_0\sigma_3}.
\]
Thus continuation around the origin replaces $\log\lambda$ by
$\log\lambda+2\pi i$ and multiplies $E$ on the right by
$F_\infty=e^{-2\pi i\theta_\infty\sigma_3}$, whereas the local Fuchsian
factor acquires $e^{2\pi i\theta_0\sigma_3}$. The formula for $\Psi^{(0)}$
is written for the generic nonresonant case; standard logarithmic terms may
occur when $2\theta_0\in\mathbb Z$. For integral exponents the power factor
may become single-valued, but the singular point of the system at
$\lambda=0$ remains.

The initial solution of the Volterra equation is constructed in
$|\lambda|>R$. Under analytic continuation toward the interior, the
solutions lie on the universal cover of $\mathbb C\setminus\{0\}$, do not
cross the selected cut, and retain their homotopy classes. A path cannot be
dragged through the origin, while a deformation that encircles the origin
introduces the local monodromy $M_0$. In the formula for $s_4$, the matrix
$Z_5$ lies on the next sheet, and the transition is accounted for by
$F_\infty$.

\paragraph{Domains and convergence.}
Here is a precise choice of paths for $\mathcal K_k$. Fix $x$, uniformly
in a compact pole-free set when needed, and put $w=\Theta_{IV}(\lambda,x)$.
For sufficiently large $H$, choose inverse branches $\lambda=\Lambda_k(w)$
on
\[
 V_k=\{w:\sigma_k\operatorname{Im}w>H\},\quad
 \sigma_k=(-1)^{k+1},\quad \Omega_k=\Lambda_k(V_k).
\]
The asymptotic lift has
$(1-k)\pi/2<\arg\lambda<(2-k)\pi/2$, $k=1,\ldots,5$.
In particular, $\Omega_1$ lies between directions $0$ and $\pi/2$,
and $\Omega_5$ is its lift after a clockwise turn.
Inversion near the corresponding branch of $\sqrt{2w}$ gives
\[
 |\Lambda_k(w)|\asymp |w|^{1/2},\qquad
 |\Lambda_k'(w)|\le C|w|^{-1/2}.
\]
For each $w\in V_k$, choose the corresponding root $\zeta_k(w)$ of
$2w$. On a circle $|\lambda-\zeta_k(w)|=L$ with fixed sufficiently
large $L$, the term $(\lambda^2-\zeta_k(w)^2)/2$ dominates
$x\lambda-\theta_\infty\log\lambda$, uniformly for $x$ in a compact
set, once $H$ is large. Rouch\'e's theorem gives exactly one root
in that disk. The roots patch to an analytic inverse on $V_k$;
$\Theta_{IV,\lambda}=\lambda+O(1)$ then gives the derivative bound.

Fix $0<\eta<1$. For entries $12$, $11$, and $22$ lift
$w+t(1+i\sigma_k\eta)$, with $t$ decreasing from $+\infty$ to zero.
For entry $21$ lift $w+t(-1+i\sigma_k\eta)$ with the same orientation.
Both paths remain in $V_k$, including when an endpoint varies. Their
off-diagonal kernel moduli equal $e^{-2t}$, and the diagonal kernels
equal one. Each limiting direction lies strictly inside the corresponding
decay sector. For $\delta<\pi/8$ these directions belong to the
subsectors specified above. Their projections can be deformed to the
central-ray models in the figures after analytic continuation.

Let $X_k$ be the Banach space of bounded analytic matrix functions on
$\Omega_k$, with the supremum of the maximum-entry norm.
From $\mathcal R=O(\lambda^{-2})$ one obtains
\[
 \|\mathcal K_k Z\|_{X_k}
 \le C_\eta\|Z\|_{X_k}\sup_{w\in V_k}
       \int_0^\infty (|w|+t)^{-3/2}\,dt
 \le 2C_\eta H^{-1/2}\|Z\|_{X_k}.
\]
Here $|w+t(\pm1+i\sigma_k\eta)|\ge c_\eta(|w|+t)$ because
$w$ and the direction vector lie in the same open half-plane and their
angle stays below $\pi-\arctan\eta$. Choose $2C_\eta H^{-1/2}<1$.
The Neumann series converges in $X_k$ and determines a unique bounded
solution of the integral equations. Locally uniform convergence permits
differentiation and recovers the spectral equation. On closed subsectors,
the same bounds give $Z_k-I=O(|\lambda|^{-1})$.
Repeating the construction with higher formal truncations gives the full
formal asymptotic expansion. The bounded normalized solution is unique:
a nontrivial constant triangular change would grow exponentially in one
of the two parts of $\Omega_k$ where $\operatorname{Re}w$ has opposite
signs. In neighboring domains the common subdominant column is therefore
the same; this yields the triangular matrices in
(\ref{eqP4StokesMatricesArticle}).
Differential-equation continuation defines the canonical solutions at common
finite comparison points on the cover. Thus both the endpoint choices and
the solutions are fixed by this construction.

With these choices, the Volterra equation is solved by the
convergent Neumann series
\[
    Z_k
    =
    I+\mathcal K_kI+\mathcal K_k^2I+\cdots,
\]
and the solution is then analytically continued into the sector.

Choose $\lambda_*$ in the intersection of the domains of analyticity of
$\Psi_k$ and $\Psi_{k+1}$, and set $E_*=E(\lambda_*,x)$. Equation
(\ref{eqP4StokesConnectionArticle}) gives directly
\[
    S_k
    =
    E_*^{-1}
    Z_k(\lambda_*)^{-1}Z_{k+1}(\lambda_*)E_*.
\]
The right-hand side is independent of $\lambda_*$ because both $\Psi_k$ and
$\Psi_{k+1}$ satisfy the same spectral system. For $k=4$, $Z_5$ denotes the
continuation to the next sheet of the universal cover.

The integral formulas are
\[
\begin{aligned}
    s_1
    &=
    e^{2\Theta_{IV}(\lambda_*)}
    \bigl(Z_1^{-1}Z_2\bigr)_{21}(\lambda_*),&
    s_2
    &=
    e^{-2\Theta_{IV}(\lambda_*)}
    \bigl(Z_2^{-1}Z_3\bigr)_{12}(\lambda_*),\\
    s_3
    &=
    e^{2\Theta_{IV}(\lambda_*)}
    \bigl(Z_3^{-1}Z_4\bigr)_{21}(\lambda_*),&
    s_4
    &=
    e^{-2\Theta_{IV}(\lambda_*)}
    \bigl(Z_4^{-1}Z_5\bigr)_{12}(\lambda_*),
\end{aligned}
\]
where each matrix $Z_k$ is determined by its Volterra equation. At
$\lambda_*$, set
\[
    d_k=\det Z_k
    =z_{11}^{(k)}z_{22}^{(k)}-z_{12}^{(k)}z_{21}^{(k)}.
\]
Then the same four identities take the fully componentwise form
\[
\begin{aligned}
s_1&=\frac{e^{2\Theta_{IV}(\lambda_*)}}{d_1}
 \left(-z_{21}^{(1)}z_{11}^{(2)}
       +z_{11}^{(1)}z_{21}^{(2)}\right)(\lambda_*),\\
s_2&=\frac{e^{-2\Theta_{IV}(\lambda_*)}}{d_2}
 \left(z_{22}^{(2)}z_{12}^{(3)}
       -z_{12}^{(2)}z_{22}^{(3)}\right)(\lambda_*),\\
s_3&=\frac{e^{2\Theta_{IV}(\lambda_*)}}{d_3}
 \left(-z_{21}^{(3)}z_{11}^{(4)}
       +z_{11}^{(3)}z_{21}^{(4)}\right)(\lambda_*),\\
s_4&=\frac{e^{-2\Theta_{IV}(\lambda_*)}}{d_4}
 \left(z_{22}^{(4)}z_{12}^{(5)}
       -z_{12}^{(4)}z_{22}^{(5)}\right)(\lambda_*).
\end{aligned}
\]
Substitution of the Neumann series turns these expressions into convergent
series of iterated integrals of the entries of $\mathcal R$.

The same formulas can be written directly in terms of the columns
$\psi_k^{(1)}$ and $\psi_k^{(2)}$ of $\Psi_k$. Since $\det\Psi_k=1$, the
triangular form of (\ref{eqP4StokesMatricesArticle}) gives
\[
    s_k
    =
    \det\bigl(\psi_k^{(1)},\psi_{k+1}^{(1)}\bigr),
    \qquad k=1,3,
\]
\[
    s_k
    =
    \det\bigl(\psi_{k+1}^{(2)},\psi_k^{(2)}\bigr),
    \qquad k=2,4.
\]
The columns in these determinants solve the integral equations given above.

Before extracting the first term of the Neumann series, fix the orientation
of the contour differences. All four contours
$\Gamma_k^{12}$, $\Gamma_{k+1}^{12}$,
$\Gamma_k^{21}$, and $\Gamma_{k+1}^{21}$ run along their respective
admissible rays from large
$|\lambda|$ to the same point $\lambda_*$. Therefore,
\[
\begin{aligned}
    \mathcal C_k^{12}
    &:={\Gamma}_{k+1}^{12}-{\Gamma}_{k}^{12}:
    \quad \infty_{k+1}^{(+)}\longrightarrow\lambda_*
    \longrightarrow\infty_k^{(+)},\\
    \mathcal C_k^{21}
    &:={\Gamma}_{k+1}^{21}-{\Gamma}_{k}^{21}:
    \quad \infty_{k+1}^{(-)}\longrightarrow\lambda_*
    \longrightarrow\infty_k^{(-)}.
\end{aligned}
\]
The sign $(+)$ means $\operatorname{Re}\Theta_{IV}\to+\infty$ and refers to
the $12$ component, while $(-)$ means
$\operatorname{Re}\Theta_{IV}\to-\infty$ and refers to the $21$ component.
For the active entries ($12$ for even $k$, $21$ for odd $k$),
the points $\infty_{k+1}^{(\pm)}$ and $\infty_k^{(\pm)}$ lie in
\emph{distinct} connected Stokes sectors on the universal cover; within one
component, the difference of homotopic paths vanishes by Cauchy's theorem.
On the second half of each composite path, the original orientation of
$\Gamma_k^{ij}$ is reversed by the minus sign. Figure
\ref{figP4PairedContoursArticle} shows models for $k=2,4$ in the $12$ case
and $k=1,3$ in the $21$ case; the second value of $k$ is obtained by rotating
the picture clockwise through $\pi$. In the left model ($k=2$),
$\arg\mu$ tends to $0$ on $\Gamma_k^{12}$ and to $-\pi$ on
$\Gamma_{k+1}^{12}$. In the right model ($k=1$), the corresponding
directions are $\pi/2$ and $-\pi/2$.
For $k=4$, the contour $\Gamma_5^{12}$ lies on
the next sheet of the universal cover.

\begin{figure}[htbp]
\centering
\begin{tikzpicture}[>=Latex,font=\scriptsize]
\begin{scope}[xshift=2.8cm,scale=0.72]
    \fill[black!5] (0,0)--(-45:3.35) arc (-45:45:3.35)--cycle;
    \fill[black!5] (0,0)--(135:3.35) arc (135:225:3.35)--cycle;
    \draw[->] (-3.65,0)--(3.65,0) node[right] {$\operatorname{Re}\mu$};
    \draw[->] (0,-2.65)--(0,2.75) node[above] {$\operatorname{Im}\mu$};
    \foreach \ang in {45,135,225,315}
        \draw[black,densely dotted] (0,0)--(\ang:3.25);
    \draw[gray,dash dot] (0,0) circle (2.85);
    \draw[black,very thick,double=white,double distance=0.8pt]
        (-3.45,0)--(-0.12,0);
    \node[above,align=center] at (-2.05,0.05)
        {cut; the tail follows\\one of its banks};
    \fill[white] (0,0) circle (0.09);
    \draw (0,0) circle (0.09) node[below right] {$0$};
    \coordinate (P12) at (0.95,0.63);
    \fill (P12) circle (1.7pt) node[above right] {$\lambda_*$};
    \draw[black,very thick,densely dashed,->]
        (3.45,0.22) .. controls (2.65,0.22) and (1.65,0.38) .. (P12);
    \node[below] at (2.35,0.19)
        {$\Gamma_k^{12}$};
    \node[right] at (3.38,0.42)
        {$\infty_k^{(+)}$};
    \draw[black,very thick,->]
        (-3.42,-0.28) .. controls (-2.15,-0.55) and (-0.35,-1.35) .. (P12);
    \node[below] at (-1.72,-0.76)
        {$\Gamma_{k+1}^{12}$};
    \node[left] at (-3.34,-0.48)
        {$\infty_{k+1}^{(+)}$};
    \node[fill=white,inner sep=1pt] at (0,-3.20)
        {$\mathcal C_k^{12}:\ \infty_{k+1}^{(+)}\to\lambda_*
          \to\infty_k^{(+)}$};
    \node[font=\small] at (0,3.50) {$12$ component, $k=2,4$};
\end{scope}

\begin{scope}[xshift=9.0cm,scale=0.72]
    \fill[black!14] (0,0)--(45:3.35) arc (45:135:3.35)--cycle;
    \fill[black!14] (0,0)--(225:3.35) arc (225:315:3.35)--cycle;
    \draw[->] (-2.75,0)--(2.85,0) node[right] {$\operatorname{Re}\mu$};
    \draw[->] (0,-3.65)--(0,3.65) node[above left] {$\operatorname{Im}\mu$};
    \foreach \ang in {45,135,225,315}
        \draw[black,densely dotted] (0,0)--(\ang:3.25);
    \draw[gray,dash dot] (0,0) circle (2.85);
    \draw[black,very thick,double=white,double distance=0.8pt]
        (-2.70,0)--(-0.12,0);
    \fill[white] (0,0) circle (0.09);
    \draw (0,0) circle (0.09) node[below right] {$0$};
    \coordinate (P21) at (0.82,0.52);
    \fill (P21) circle (1.7pt) node[above right] {$\lambda_*$};
    \draw[black,very thick,densely dashed,->]
        (0.20,3.45) .. controls (0.20,2.35) and (0.48,1.05) .. (P21);
    \node[right] at (0.28,2.27)
        {$\Gamma_k^{21}$};
    \node[right] at (0.50,3.35)
        {$\infty_k^{(-)}$};
    \draw[black,very thick,->]
        (-0.18,-3.45) .. controls (-0.10,-2.15) and (0.22,-0.45) .. (P21);
    \node[left] at (-0.06,-2.10)
        {$\Gamma_{k+1}^{21}$};
    \node[left] at (-0.50,-3.35)
        {$\infty_{k+1}^{(-)}$};
    \node[fill=white,inner sep=1pt] at (0,-4.05)
        {$\mathcal C_k^{21}:\ \infty_{k+1}^{(-)}\to\lambda_*
          \to\infty_k^{(-)}$};
    \node[font=\small] at (0,4.55) {$21$ component, $k=1,3$};
\end{scope}
\end{tikzpicture}
\caption{Pairs of contours entering the differences
$\Gamma_{k+1}^{12}-\Gamma_k^{12}$ and
$\Gamma_{k+1}^{21}-\Gamma_k^{21}$. The solid and dashed contours are
oriented from large $|\lambda|$ along their asymptotic rays toward
$\lambda_*$. In the oriented difference, the dashed contour is traversed in
the reverse direction. The two asymptotic branches of each composite path
lie in distinct connected Stokes sectors. The dash-dotted curve is
$|\mu|=R$, dotted lines are Stokes rays, and the double line is the selected
cut.}
\label{figP4PairedContoursArticle}
\end{figure}
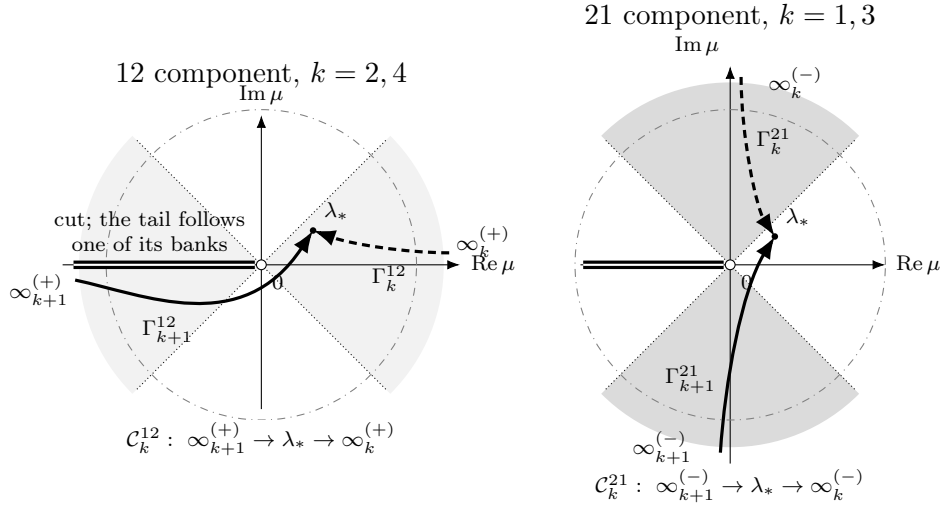

Denote the sum of the Neumann-series terms containing two or more integrations
by $\mathcal N_k^{(\geq2)}$. The single-integral terms are
\[
\begin{aligned}
    s_k
    &=
    \left(
        \int_{\Gamma_{k+1}^{21}}
        -
        \int_{\Gamma_k^{21}}
    \right)
    e^{2\Theta_{IV}(\mu)}
    \mathcal R_{21}(\mu)\,d\mu
    +\mathcal N_k^{(\geq2)},\\
    &\qquad k=1,3,
\end{aligned}
\]
\[
\begin{aligned}
    s_k
    &=
    \left(
        \int_{\Gamma_{k+1}^{12}}
        -
        \int_{\Gamma_k^{12}}
    \right)
    e^{-2\Theta_{IV}(\mu)}
    \mathcal R_{12}(\mu)\,d\mu
    +\mathcal N_k^{(\geq2)},\\
    &\qquad k=2,4.
\end{aligned}
\]
The contour difference is understood on the universal cover. The single
integrals form the first term of the convergent Neumann series, and
$\mathcal N_k^{(\geq2)}$ is generated by repeated applications of
$\mathcal K_k$.

\section{The squared-eigenfunction identity}

Let $\psi=(\psi_1,\psi_2)^T$ be any column of a fundamental solution of the
Lax pair, and set
\[
    W_1=\psi_1^2,\qquad
    W_2=\psi_1\psi_2,\qquad
    W_3=\psi_2^2.
\]
Introduce the abbreviations
\[
    d=\lambda+x+\frac q2,
    \qquad
    \chi=qh+2\theta_\infty.
\]
The $x$-part of the Lax pair gives
\[
    \partial_x\psi_1=d\psi_1+\psi_2,
    \qquad
    \partial_x\psi_2=-\chi\psi_1-d\psi_2,
\]
and hence
\[
\begin{aligned}
    \partial_xW_1&=2dW_1+2W_2,\\
    \partial_xW_2&=-\chi W_1+W_3,\\
    \partial_xW_3&=-2\chi W_2-2dW_3.
\end{aligned}
\]
Similarly, the spectral system with the coefficients $a$, $b$, and $c$ from
(\ref{eqP4aArticle})--(\ref{eqP4cArticle}) gives
\[
\begin{aligned}
    \partial_\lambda W_1&=2aW_1+2bW_2,\\
    \partial_\lambda W_2&=cW_1+bW_3,\\
    \partial_\lambda W_3&=2cW_2-2aW_3.
\end{aligned}
\]

Fix the same branch of $\sqrt q$ as in
(\ref{eqP4SelfAdjointSubstitutionArticle}), and set
\[
\begin{aligned}
    Q_{IV}
    &=\frac{\sqrt q}{2\lambda}
      \bigl((h-\lambda)W_1-W_2\bigr),\\
    R_{IV}
    &=-\sqrt q\left(\lambda+2x+\frac{3q}{2}\right)W_1.
\end{aligned}
\]

\begin{proposition}
For every column solution of the Lax pair, the quadratic expressions defined
above satisfy the identity
\begin{equation}
    \frac{d^2}{dx^2}Q_{IV}
    =
    \left(
        x^2-\alpha
        +6xq
        +\frac{15}{4}q^2
        -\frac{3\beta}{2q^2}
    \right)Q_{IV}
    +
    \frac{d}{d\lambda}R_{IV}.
    \label{eqP4SquaredEigenfunctionIdentityArticle}
\end{equation}
\end{proposition}

\begin{proof}
The equations for $q$ and $h$ give
\[
    (\sqrt q)'=\left(h+x+\frac q2\right)\sqrt q,
    \qquad
    h'=-h^2-2h(q+x)-2\theta_\infty
       -\frac{4\theta_0^2}{q^2}.
\]
Set
\[
    \zeta=x+\frac q2,
    \qquad
    L=\lambda+2x+\frac{3q}{2}.
\]
Differentiating the definition of $Q_{IV}$ once and substituting the
$x$-system for $W_j$ give
\[
\begin{aligned}
    \partial_xQ_{IV}
    =\frac{\sqrt q}{2\lambda}\Bigg[&
      \left(
        h(\lambda+\zeta)-2\lambda^2-3\zeta\lambda
        -\frac{4\theta_0^2}{q^2}
      \right)W_1\\
      &+(h-\zeta-2\lambda)W_2-W_3\Bigg].
\end{aligned}
\]
Differentiating once more, using the three $x$-equations again, and then
substituting the evolution equations for $q$ and $h$, we obtain after
collecting coefficients
\[
\begin{aligned}
    \partial_x^2Q_{IV}-U_{IV}Q_{IV}
    =-\sqrt q\Bigl(&[1+2La]W_1+2LbW_2\Bigr),
\end{aligned}
\]
where
\[
    U_{IV}=x^2-\alpha+6xq+\frac{15}{4}q^2
           -\frac{3\beta}{2q^2}.
\]
The coefficient of $W_3$ in this difference vanishes identically. On the
other hand, when $R_{IV}$ is differentiated with respect to $\lambda$, the
quantities $x$, $q$, and $h$ are held fixed and $\partial_\lambda L=1$.
Therefore, the first spectral equation for the squares gives
\[
\begin{aligned}
    \partial_\lambda R_{IV}
    &=-\sqrt q\left(W_1+L\partial_\lambda W_1\right)\\
    &=-\sqrt q\Bigl([1+2La]W_1+2LbW_2\Bigr).
\end{aligned}
\]
The right-hand sides of the last two identities coincide, proving
(\ref{eqP4SquaredEigenfunctionIdentityArticle}).
\end{proof}

\subsection*{Existence of the contour $\Gamma$}

\begin{lemma}[Cauchy obstruction]
Let $\widetilde{\mathbb C^*}$ denote the universal cover of
$\mathbb C^*=\mathbb C\setminus\{0\}$, and let $\mathcal D$ be a connected
component of its rapid-decay region. Let
$\gamma:\infty_a\to\infty_b$ be a path with two asymptotic branches in
$\mathcal D$, where $\infty_a$ and $\infty_b$ denote the corresponding
lifted directions as $|\lambda|\to\infty$, and let $u(\lambda,x)$ be a
column solution of the Lax pair continued along $\gamma$. Suppose that the
form $Q_{IV}(u(\lambda,x))\,d\lambda$ is holomorphic in a simply connected
domain $U$ containing $\gamma$ and closing arcs of large radius inside
$\mathcal D$, and that it tends to zero uniformly on these arcs so that the
arc integral tends to zero. Then
\begin{equation}
    \int_\gamma Q_{IV}(u(\lambda,x))\,d\lambda=0.
    \label{eqP4SameSectorZeroArticle}
\end{equation}
\end{lemma}

\begin{proof}
Truncate both tails at the circle $|\lambda|=\rho$ and denote the resulting
finite path by $\gamma_\rho$. Since its endpoints belong to the same component
$\mathcal D$, they can be joined by an arc $A_\rho\subset\mathcal D$ such
that the closed contour $\gamma_\rho+A_\rho$ is null-homologous in $U$.
Cauchy's theorem gives
\[
    \int_{\gamma_\rho}Q_{IV}\,d\lambda
    =-\int_{A_\rho}Q_{IV}\,d\lambda.
\]
The right-hand side tends to zero as $\rho\to\infty$, whereas the left-hand
side tends to the defining improper integral. This proves
(\ref{eqP4SameSectorZeroArticle}).
\end{proof}

A nontrivial two-ended contour has tails in distinct connected components of
the rapid-decay region on the universal cover. Different rays in the same
component do not count as distinct sectors.

Subdominant columns in different sectors are related by Stokes matrices.
Accordingly, the boundary term is removed by a rapid-decay cycle with
coefficients in the symmetric square of the local system defined by the Lax
pair.

Represent the universal cover as
\[
    \widetilde{\mathbb C^*}
    =\{(r,\varphi):r>0,\ \varphi\in\mathbb R\},
    \qquad
    \pi(r,\varphi)=re^{i\varphi}.
\]
On it, $\log\lambda=\log r+i\varphi$ is single-valued. Choose the four lifted
directions
\[
    \varphi_0=0,\qquad
    \varphi_1=\frac{\pi}{2},\qquad
    \varphi_2=\pi,\qquad
    \varphi_3=-\frac{\pi}{2}.
\]
For $j=0,1,2,3$, set
\begin{equation}
    \mathcal D_j(\delta)
    =\left\{(r,\varphi):
       \left|\varphi-\varphi_j\right|
       <\frac{\pi}{4}-\delta\right\},
    \qquad 0<\delta<\frac{\pi}{4}.
    \label{eqP4DecaySectorsArticle}
\end{equation}
For odd $j$, $\operatorname{Re}\Theta_{IV}\to-\infty$, whereas for even
$j$, $\operatorname{Re}\Theta_{IV}\to+\infty$. Denote by $u_j$ the
subdominant column of the jointly normalized matrix $\Phi_k$, continued
into $\mathcal D_j(\delta)$. In every closed
subsector it has the uniform asymptotic behavior
\begin{equation}
u_j\sim
\begin{cases}
\mathfrak d^{-1}e^{-\Theta_{IV}}
\bigl(O(\lambda^{-1}),\ 1+O(\lambda^{-1})\bigr)^T,
&j=0,2,\\[2mm]
\mathfrak d e^{\Theta_{IV}}
\bigl(1+O(\lambda^{-1}),\ O(\lambda^{-1})\bigr)^T,
&j=1,3.
\end{cases}
\label{eqP4SubdominantColumnsArticle}
\end{equation}

Choose $R>0$ and the point $\lambda_*=(R,0)$ on the cover. These data and
the paths below are independent of $x$. For each $j$, introduce the central
ray of $\mathcal D_j(\delta)$,
\[
    \ell_j=\{(r,\varphi_j):r\geq R\}
    \subset\mathcal D_j(\delta).
\]
The path $\gamma_j$ is oriented from large $r$ along $\ell_j$ toward
$\lambda_*$ and consists of two parts:
\begin{equation}
\begin{aligned}
    \gamma_j^{\infty}(r)&=(r,\varphi_j),
       &&+\infty>r\ge R,\\
    \gamma_j^0(t)&=(R,(1-t)\varphi_j),
       &&0\le t\le1,\\
    \gamma_j&=\gamma_j^{\infty}*\gamma_j^0.
\end{aligned}
\label{eqP4ExplicitGammaLegsArticle}
\end{equation}
For $j=0$, the second part degenerates to a point. All paths remain in
$r\ge R$ and therefore avoid the Fuchsian point $\lambda=0$. The column
$u_j$, initially defined on the tail in $\mathcal D_j(\delta)$, has a unique
analytic continuation along $\gamma_j$ to $\lambda_*$. The projections
of these paths and their lifts are shown in
Figure~\ref{figP4RapidDecayCycleArticle}.

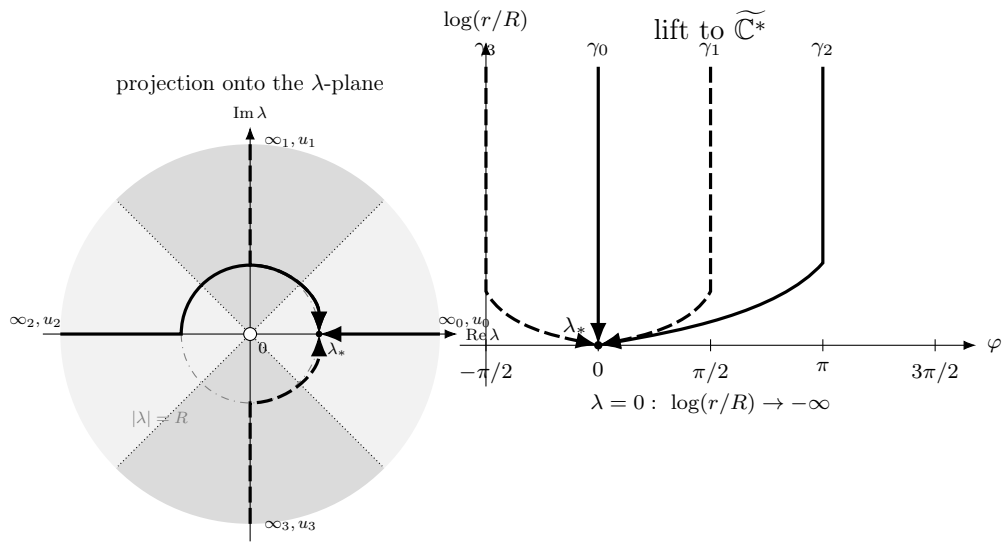
\begin{figure}[htbp]
\centering
\begin{tikzpicture}[>=Latex,font=\scriptsize,scale=0.99,transform shape]
\begin{scope}[scale=0.78]
    \fill[black!5] (0,0)--(-45:3.25) arc (-45:45:3.25)--cycle;
    \fill[black!14] (0,0)--(45:3.25) arc (45:135:3.25)--cycle;
    \fill[black!5] (0,0)--(135:3.25) arc (135:225:3.25)--cycle;
    \fill[black!14] (0,0)--(225:3.25) arc (225:315:3.25)--cycle;
    \draw[->] (-3.55,0)--(3.55,0) node[right] {$\operatorname{Re}\lambda$};
    \draw[->] (0,-3.55)--(0,3.55) node[above] {$\operatorname{Im}\lambda$};
    \foreach \ang in {45,135,225,315}
        \draw[black,densely dotted] (0,0)--(\ang:3.20);
    \draw[gray,dash dot] (0,0) circle (1.18);
    \fill[white] (0,0) circle (0.11);
    \draw (0,0) circle (0.11) node[below right] {$0$};
    \coordinate (P) at (1.18,0);
    \fill (P) circle (1.6pt) node[below right] {$\lambda_*$};
    \draw[black,very thick,->] (3.25,0)--(P);
    \draw[black,very thick,dash pattern=on 6pt off 2pt,->]
        (0,3.25)--(0,1.18) arc (90:0:1.18);
    \draw[black,very thick,->]
        (-3.25,0)--(-1.18,0) arc (180:0:1.18);
    \draw[black,very thick,dash pattern=on 6pt off 2pt,->]
        (0,-3.25)--(0,-1.18) arc (-90:0:1.18);
    \node[right] at (3.08,0.20) {$\infty_0,u_0$};
    \node[above right] at (0.10,3.05) {$\infty_1,u_1$};
    \node[left] at (-3.08,0.20) {$\infty_2,u_2$};
    \node[below right] at (0.10,-3.05) {$\infty_3,u_3$};
    \node[gray] at (-1.55,-1.45) {$|\lambda|=R$};
    \node[font=\small] at (0,4.28) {projection onto the $\lambda$-plane};
\end{scope}
\begin{scope}[xshift=3.15cm,yshift=-0.15cm]
    \draw[->] (-0.35,0)--(6.55,0) node[right] {$\varphi$};
    \draw[->] (0,-0.55)--(0,4.05) node[above] {$\log(r/R)$};
    \node[align=center] at (3.0,-0.78)
        {$\lambda=0:\ \log(r/R)\to-\infty$};
    \draw (0,0.08)--(0,-0.08) node[below] {$-\pi/2$};
    \draw (1.5,0.08)--(1.5,-0.08) node[below] {$0$};
    \draw (3,0.08)--(3,-0.08) node[below] {$\pi/2$};
    \draw (4.5,0.08)--(4.5,-0.08) node[below] {$\pi$};
    \draw (6,0.08)--(6,-0.08) node[below] {$3\pi/2$};
    \coordinate (Ps) at (1.5,0);
    \fill (Ps) circle (1.6pt) node[above left] {$\lambda_*$};
    \draw[black,very thick,dash pattern=on 6pt off 2pt,->]
        (0,3.72)--(0,0.72) .. controls (0.25,0.25) and (0.95,0.12) .. (Ps);
    \draw[black,very thick,->] (1.5,3.72)--(Ps);
    \draw[black,very thick,dash pattern=on 6pt off 2pt,->]
        (3,3.72)--(3,0.72) .. controls (2.75,0.25) and (2.05,0.12) .. (Ps);
    \draw[black,very thick,->]
        (4.5,3.72)--(4.5,1.10) .. controls (3.90,0.42) and (2.50,0.18) .. (Ps);
    \node[above] at (0,3.72) {$\gamma_3$};
    \node[above] at (1.5,3.72) {$\gamma_0$};
    \node[above] at (3,3.72) {$\gamma_1$};
    \node[above] at (4.5,3.72) {$\gamma_2$};
    \node[font=\small] at (3.0,4.28)
        {lift to $\widetilde{\mathbb C^*}$};
\end{scope}
\end{tikzpicture}
\caption{The rapid-decay chain $\Gamma$. The left panel shows four oriented
branches $\gamma_j:\infty_j\to\lambda_*$, with the points $\infty_j$ in
pairwise distinct Stokes sectors. Light-gray sectors correspond to
$\operatorname{Re}\Theta_{IV}>0$ and the columns $u_0,u_2$; dark-gray
sectors correspond to $\operatorname{Re}\Theta_{IV}<0$ and the columns
$u_1,u_3$. Solid curves represent $\gamma_0,\gamma_2$, and dashed curves
represent $\gamma_1,\gamma_3$. The right panel shows the same chain on the
unfolded universal cover; all its points stay away from $\lambda=0$.}
\label{figP4RapidDecayCycleArticle}
\end{figure}

For a column $u=(u_1,u_2)^T$, introduce its symmetric square
\[
    u^{\odot2}=(u_1^2,u_1u_2,u_2^2)^T.
\]
\begin{samepage}
\begin{lemma}[Existence of a rapid-decay cycle]
There exist constants $c_0,c_1,c_2,c_3$, not all zero and depending only on
the monodromy data, such that the chain
\begin{equation}
    \Gamma=\sum_{j=0}^3c_j(\gamma_j,u_j^{\odot2})
    \label{eqP4RapidDecayCycleArticle}
\end{equation}
satisfies
\begin{equation}
    R_{IV}\big|_{\partial\Gamma}=0.
    \label{eqP4RapidDecayBoundaryArticle}
\end{equation}
More precisely, this boundary condition has two parts. Each asymptotic branch
$\gamma_j^\infty$ lies in the explicitly specified sector and approaches
the point at infinity along its central ray as $r\to+\infty$; moreover,
\[
    \lim_{r\to+\infty}
    R_{IV}\bigl(u_j(re^{i\varphi_j},x)\bigr)=0,
    \qquad j=0,1,2,3,
\]
where the sector--direction pairs are
\[
\begin{array}{c|c|c}
 j&\text{sector on }\widetilde{\mathbb C^*}&
      \text{asymptotic direction}\ \\ \hline
 0&\mathcal D_0(\delta)&\arg\lambda=0,\\
 1&\mathcal D_1(\delta)&\arg\lambda=\pi/2\\
 2&\mathcal D_2(\delta)&\arg\lambda=\pi\\
 3&\mathcal D_3(\delta)&\arg\lambda=-\pi/2.
\end{array}
\]
At the common finite vertex of the chain, one requires
\[
    \sum_{j=0}^3c_jR_{IV}(u_j)(\lambda_*,x)=0.
\]
The integral
\begin{equation}
    \int_\Gamma Q_{IV}\,d\lambda
    :=\sum_{j=0}^3c_j
      \int_{\gamma_j}Q_{IV}(u_j(\lambda,x))\,d\lambda
    \label{eqP4RapidDecayIntegralDefinitionArticle}
\end{equation}
converges absolutely and locally uniformly in $x$ wherever $q(x)\ne0$.
It may be differentiated twice with respect to $x$ under the integral sign.
Furthermore, every nonzero chain of this type has at least two nonzero
coefficients $c_j$; hence its support contains asymptotic branches in at
least two distinct sectors $\mathcal D_j(\delta)$.
\end{lemma}
\end{samepage}

\begin{proof}
The vectors $u_j^{\odot2}$ belong to the three-dimensional space
$\operatorname{Sym}^2\mathbb C^2$. Thus the four vectors are linearly
dependent. A nonzero linear relation cannot contain exactly one nonzero
coefficient because $u_j^{\odot2}\ne0$. Its support therefore contains at
least two indices, and the corresponding asymptotic branches lie in distinct
sectors $\mathcal D_j(\delta)$. The coefficients may be chosen as
\begin{equation}
    c_j=(-1)^j
    \det\bigl(w_0,\ldots,\widehat w_j,\ldots,w_3\bigr),
    \qquad
    w_j=u_j^{\odot2}(\lambda_*,x_0),
    \label{eqP4RapidDecayCoefficientsArticle}
\end{equation}
where $x_0$ is fixed. If all minors in
(\ref{eqP4RapidDecayCoefficientsArticle}) vanish, then the rank of the family
$w_j$ is less than three, and one chooses any nonzero vector in the kernel of
the matrix $(w_0,w_1,w_2,w_3)$. Moreover,
\[
\det(u_a^{\odot2},u_b^{\odot2},u_c^{\odot2})
=\det(u_a,u_b)\det(u_a,u_c)\det(u_b,u_c),
\]
so, upon setting $\Delta_{ab}=\det(u_a,u_b)$, we obtain
\[
\begin{aligned}
 c_0&=\Delta_{12}\Delta_{13}\Delta_{23},&
 c_1&=-\Delta_{02}\Delta_{03}\Delta_{23},\\
 c_2&=\Delta_{01}\Delta_{03}\Delta_{13},&
 c_3&=-\Delta_{01}\Delta_{02}\Delta_{12}.
\end{aligned}
\]
These are Wronskians of jointly normalized columns and functions of the
reference monodromy data. The factors $\mathfrak d^{\pm1}$ in the column asymptotics
are essential for their constancy. Since $\operatorname{tr}A_{IV}=\operatorname{tr}B_{IV}=0$,
we have $\partial_\lambda\Delta_{ab}=\partial_x\Delta_{ab}=0$.
Consequently, the coefficients $c_j$ are independent of both $\lambda$ and
$x$.

Each $u_j^{\odot2}$ satisfies the same symmetric-square system. Hence the
relation
\begin{equation}
    \sum_{j=0}^3c_j u_j^{\odot2}(\lambda,x)=0,
    \label{eqP4SymmetricSquareRelationArticle}
\end{equation}
imposed at $(\lambda_*,x_0)$ is preserved under continuation in $\lambda$ and
$x$. At the common finite vertex $\lambda_*$ it gives
\[
    \sum_{j=0}^3c_jR_{IV}(u_j)(\lambda_*,x)=0.
\]
Along each branch $\gamma_j^\infty\subset\mathcal D_j(\delta)$, with
$\lambda=re^{i\varphi_j}$ and $r\to+\infty$,
(\ref{eqP4SubdominantColumnsArticle}) implies
$R_{IV}(u_j(re^{i\varphi_j},x))\to0$. This proves
(\ref{eqP4RapidDecayBoundaryArticle}).

On the central ray of $\mathcal D_j(\delta)$,
$\cos(2\varphi_j)=(-1)^j$. If $x$ ranges over a compact set $K$, then for
sufficiently large $r$,
\begin{equation}
\begin{aligned}
 \operatorname{Re}\Theta_{IV}(re^{i\varphi_j},x)
 &\ge \frac14r^2,&&j=0,2,\\
 \operatorname{Re}\Theta_{IV}(re^{i\varphi_j},x)
 &\le-\frac14r^2,&&j=1,3,
\end{aligned}
\qquad x\in K.
\label{eqP4RapidDecayPhaseEstimateArticle}
\end{equation}
The linear term $x\lambda$ and the logarithmic term in the phase do not alter
these estimates, because they are of orders $O_K(r)$ and $O(\log r)$.
Equations (\ref{eqP4SubdominantColumnsArticle}) and
(\ref{eqP4RapidDecayPhaseEstimateArticle}) yield, on every tail, the estimate
\[
    |Q_{IV}(u_j)|+|R_{IV}(u_j)|
    \le C_Kr^N e^{-r^2/2}
\]
for some $C_K,N>0$. One or two differentiations with respect to $x$ introduce
only additional powers of $r$ through the $x$-part of the Lax pair. The same
integrable majorant applies to $\partial_xQ_{IV}$ and
$\partial_x^2Q_{IV}$. The finite arcs $\gamma_j^0$ are compact and avoid
$\lambda=0$, so the integrands and their $x$-derivatives are bounded there.
The dominated convergence theorem gives absolute and locally uniform
convergence and permits two differentiations under the integral sign.
\end{proof}

In the formula below, $\Gamma$ is not an ordinary path in the $\lambda$-plane
but a chain with coefficients in the symmetric square of the local system of
Lax-pair solutions. Its branches $\gamma_j^\infty$ lie in the sectors
$\mathcal D_j(\delta)$ of the universal cover
$\widetilde{\mathbb C^*}$ and approach infinity as $r\to+\infty$ along the
rays $\arg\lambda=0,\pi/2,\pi,-\pi/2$. On $\gamma_j^\infty$, the
subdominant column $u_j$ is chosen so that
\[
    R_{IV}(u_j(re^{i\varphi_j},x))\longrightarrow0,
    \qquad r\longrightarrow+\infty.
\]
For generic $\theta_0$ and $\theta_\infty$, this cover is not two-sheeted:
the power factors $\lambda^{\theta_0\sigma_3}$ and
$\lambda^{-\theta_\infty\sigma_3}$ have arbitrary monodromy. A finite-sheeted
cover, in particular a two-sheeted one, occurs only for special rational or
half-integral exponents and is not assumed here.

Integrating (\ref{eqP4SquaredEigenfunctionIdentityArticle}) over the chain
(\ref{eqP4RapidDecayCycleArticle}) now gives
\[
\begin{aligned}
    \frac{d^2}{dx^2}\int_\Gamma Q_{IV}\,d\lambda
    &-U_{IV}(x)\int_\Gamma Q_{IV}\,d\lambda\\
    &=\sum_{j=0}^3c_j
      \int_{\gamma_j}\partial_\lambda R_{IV}(u_j)\,d\lambda
      =R_{IV}\big|_{\partial\Gamma}=0.
\end{aligned}
\]
The formula
\begin{equation}
    y(x)=\int_\Gamma Q_{IV}\,d\lambda,
    \qquad
    v(x)=\sqrt{q(x)}\,y(x)
    \label{eqP4LinearizedIntegralAnsatzArticle}
\end{equation}
defines solutions of the normal-form and original linearized equations. For
special monodromy data, the cycle may contain two or three nonzero branches;
by the Cauchy-obstruction lemma, they approach infinity in at least two
sectors.

\subsection*{A second integral solution and an independence criterion}

The integral associated with a rapid-decay homology class is its period.
To construct a second solution, choose another chain class
$\widetilde\Gamma$ that also
satisfies the boundary conditions
\[
 R_{IV}\big|_{\partial\widetilde\Gamma}=0,
 \qquad
 R_{IV}(\widetilde u_n(re^{in\pi/2},x))\longrightarrow0
 \quad(r\to+\infty)
\]
along each asymptotic branch
$\widetilde\gamma_n^\infty\subset\widehat{\mathcal D}_n(\delta)$. The chain
$\widetilde\Gamma$ cannot be obtained by a simple deformation of $\Gamma$
with fixed asymptotic directions: such a deformation preserves the
rapid-decay class and the same period. One must change the pair of lifted
asymptotic sectors or pass to the next sheet of the universal cover.

We now specify two such cycles explicitly. For $n\in\mathbb Z$, let
\[
 \widehat{\mathcal D}_n(\delta)
 =\left\{(r,\varphi):
 \left|\varphi-\frac{n\pi}{2}\right|
 <\frac{\pi}{4}-\delta\right\}
\]
be the lifted sector centered at $n\pi/2$, and let $u_n$ be the canonical
subdominant column in that sector. Choose the sectors and the asymptotic
directions of the two cycles as follows:
\[
\begin{aligned}
 \Gamma_A:\quad&
 \widehat{\mathcal D}_{-1}\ [\arg\lambda=-\pi/2],\quad
 \widehat{\mathcal D}_{0}\ [\arg\lambda=0],\\
 &\widehat{\mathcal D}_{1}\ [\arg\lambda=\pi/2],\quad
 \widehat{\mathcal D}_{2}\ [\arg\lambda=\pi];\\[1mm]
 \Gamma_B:\quad&
 \widehat{\mathcal D}_{1}\ [\arg\lambda=\pi/2],\quad
 \widehat{\mathcal D}_{2}\ [\arg\lambda=\pi],\\
 &\widehat{\mathcal D}_{3}\ [\arg\lambda=3\pi/2],\quad
 \widehat{\mathcal D}_{4}\ [\arg\lambda=2\pi].
\end{aligned}
\]
Thus the two cycles have asymptotic branches in the common sectors
$\widehat{\mathcal D}_{1}$ and $\widehat{\mathcal D}_{2}$, while
$\widehat{\mathcal D}_{-1},\widehat{\mathcal D}_{0}$ occur only in the
first cycle and $\widehat{\mathcal D}_{3},\widehat{\mathcal D}_{4}$ only in
the second.

More precisely, the asymptotic parts of the branches are
\[
 \gamma_n^{A,\infty}(r)=(r,n\pi/2),\quad n=-1,0,1,2,
 \qquad +\infty>r\geq R,
\]
\[
 \gamma_n^{B,\infty}(r)=(r,n\pi/2),\quad n=1,2,3,4,
 \qquad +\infty>r\geq R.
\]
Hence, for sufficiently large $r$, each branch lies in the explicitly
specified sector $\widehat{\mathcal D}_n(\delta)$ and approaches infinity
along its central ray.

Let the branches $\gamma_n^A$ of the first cycle meet at the vertex
$\lambda_*^A=(R,0)$ and the branches $\gamma_n^B$ of the second cycle at
$\lambda_*^B=(R,\pi)$ on the universal cover. Linear relations among the
four symmetric squares at these vertices determine the coefficients $a_n$
and $b_n$ and the chains
\[
 \Gamma_A=\sum_{n=-1}^{2}a_n(\gamma_n^A,u_n^{\odot2}),
 \qquad
 \Gamma_B=\sum_{n=1}^{4}b_n(\gamma_n^B,u_n^{\odot2}).
\]
In each cycle, the coefficients are selected by
(\ref{eqP4RapidDecayCoefficientsArticle}) applied to the corresponding four
columns. Thus the finite boundary values of $R_{IV}$ cancel, while along
every asymptotic branch specified above,
\[
 R_{IV}(u_n(re^{in\pi/2},x))\longrightarrow0,
 \qquad r\longrightarrow+\infty.
\]
Below we consider the generic case
\[
 a_{-1}a_0b_3b_4\ne0,
\]
in which each cycle contains two asymptotic sectors absent from the other.
The projections of some branches onto the $\lambda$-plane coincide, but
their lifts centered at $-\pi/2$ and $3\pi/2$, and at $0$ and $2\pi$,
belong to different sheets. The two cycles are shown in
Figure~\ref{figP4TwoIndependentCyclesArticle}.

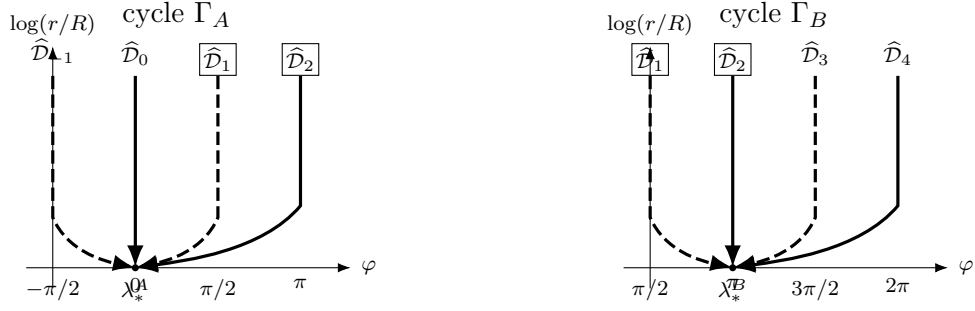
\begin{figure}[htbp]
\centering
\begin{tikzpicture}[>=Latex,font=\scriptsize]
\begin{scope}[xshift=0.2cm,scale=0.78]
    \draw[->] (-0.45,0)--(5.05,0) node[right] {$\varphi$};
    \draw[->] (0,-0.35)--(0,3.75) node[above] {$\log(r/R)$};
    \coordinate (PA) at (1.4,0);
    \fill (PA) circle (1.7pt) node[below] {$\lambda_*^A$};
    \draw[very thick,dash pattern=on 5pt off 2pt,->]
        (0,3.25)--(0,0.85) .. controls (0.25,0.28) and (0.85,0.12) .. (PA);
    \draw[very thick,->] (1.4,3.25)--(PA);
    \draw[very thick,dash pattern=on 5pt off 2pt,->]
        (2.8,3.25)--(2.8,0.85) .. controls (2.55,0.28) and (1.95,0.12) .. (PA);
    \draw[very thick,->]
        (4.2,3.25)--(4.2,1.05) .. controls (3.55,0.28) and (2.25,0.12) .. (PA);
    \node[above] at (0,3.25) {$\widehat{\mathcal D}_{-1}$};
    \node[above] at (1.4,3.25) {$\widehat{\mathcal D}_{0}$};
    \node[draw,inner sep=1.5pt,above] at (2.8,3.25)
        {$\widehat{\mathcal D}_{1}$};
    \node[draw,inner sep=1.5pt,above] at (4.2,3.25)
        {$\widehat{\mathcal D}_{2}$};
    \node[below] at (0,0) {$-\pi/2$};
    \node[below] at (1.4,0) {$0$};
    \node[below] at (2.8,0) {$\pi/2$};
    \node[below] at (4.2,0) {$\pi$};
    \node[font=\small] at (2.1,4.28) {cycle $\Gamma_A$};
\end{scope}
\begin{scope}[xshift=8.1cm,scale=0.78]
    \draw[->] (-0.45,0)--(5.05,0) node[right] {$\varphi$};
    \draw[->] (0,-0.35)--(0,3.75) node[above] {$\log(r/R)$};
    \coordinate (PB) at (1.4,0);
    \fill (PB) circle (1.7pt) node[below] {$\lambda_*^B$};
    \draw[very thick,dash pattern=on 5pt off 2pt,->]
        (0,3.25)--(0,0.85) .. controls (0.25,0.28) and (0.85,0.12) .. (PB);
    \draw[very thick,->] (1.4,3.25)--(PB);
    \draw[very thick,dash pattern=on 5pt off 2pt,->]
        (2.8,3.25)--(2.8,0.85) .. controls (2.55,0.28) and (1.95,0.12) .. (PB);
    \draw[very thick,->]
        (4.2,3.25)--(4.2,1.05) .. controls (3.55,0.28) and (2.25,0.12) .. (PB);
    \node[draw,inner sep=1.5pt,above] at (0,3.25)
        {$\widehat{\mathcal D}_{1}$};
    \node[draw,inner sep=1.5pt,above] at (1.4,3.25)
        {$\widehat{\mathcal D}_{2}$};
    \node[above] at (2.8,3.25) {$\widehat{\mathcal D}_{3}$};
    \node[above] at (4.2,3.25) {$\widehat{\mathcal D}_{4}$};
    \node[below] at (0,0) {$\pi/2$};
    \node[below] at (1.4,0) {$\pi$};
    \node[below] at (2.8,0) {$3\pi/2$};
    \node[below] at (4.2,0) {$2\pi$};
    \node[font=\small] at (2.1,4.28) {cycle $\Gamma_B$};
\end{scope}
\end{tikzpicture}
\caption{Two rapid-decay cycles on the unfolded universal cover. Boxes mark
the common sectors $\widehat{\mathcal D}_{1}$ and
$\widehat{\mathcal D}_{2}$. The first cycle has additional asymptotic
branches in $\widehat{\mathcal D}_{-1},\widehat{\mathcal D}_{0}$, and the
second in $\widehat{\mathcal D}_{3},\widehat{\mathcal D}_{4}$. Their sets
of asymptotic sectors therefore differ by at least two sectors.}
\label{figP4TwoIndependentCyclesArticle}
\end{figure}

Set $\Gamma^{(1)}=\Gamma_A$ and $\Gamma^{(2)}=\Gamma_B$, and define
\[
 y_\nu(x)=\int_{\Gamma^{(\nu)}}Q_{IV}\,d\lambda,
 \qquad \nu=1,2.
\]
Both integral solutions satisfy (\ref{eqLinP4SelfAdjointArticle}). They are linearly
independent if and only if their Wronskian at one, and hence every, point
$x_0$,
\[
\begin{aligned}
 I_\nu&=\int_{\Gamma^{(\nu)}}Q_{IV}(\lambda,x_0)\,d\lambda,\\
 J_\nu&=\int_{\Gamma^{(\nu)}}
          \partial_xQ_{IV}(\lambda,x_0)\,d\lambda,
 \qquad \nu=1,2,
\end{aligned}
\]
\[
 \mathcal W[\Gamma^{(1)},\Gamma^{(2)}]
 :=y_1(x_0)y_2'(x_0)-y_1'(x_0)y_2(x_0)
 =I_1J_2-J_1I_2
\]
is nonzero. Differentiation under the integral sign was justified in the
rapid-decay-cycle lemma. The Wronskian is independent of $x$ because
(\ref{eqLinP4SelfAdjointArticle}) contains no first derivative. Thus, to
construct a fundamental pair of solutions, it is enough to select a second cycle class,
verify the inequality at one point, and set
\[
 v_1=\sqrt q\,y_1,
 \qquad
 v_2=\sqrt q\,y_2.
\]
For $q\ne0$,
$W[v_1,v_2]=qW[y_1,y_2]\ne0$.

\begin{proposition}[Conditional generic nondegeneracy]
Let $\mathcal M$ be a connected coordinate neighborhood in the smooth part
of the monodromy manifold for fixed $\theta_0,\theta_\infty$, on which
the inverse Riemann--Hilbert problem at $x_0$ is holomorphically solvable
in the chosen gauge and $q(x_0)$ is finite. Suppose that
the fixed lifted contours $\Gamma_A,\Gamma_B$ remain admissible throughout
$\mathcal M$, that $q(x_0)\ne0$, and that the minors used to normalize the
cycle coefficients do not vanish. Then the Wronskian
$\mathcal W(\mathbf m)$ is holomorphic in the local monodromy coordinates
$\mathbf m\in\mathcal M$. If $\mathcal W(\mathbf m_*)\ne0$ at one point
$\mathbf m_*\in\mathcal M$, its zero set is a proper analytic subset
$\mathcal Z\subset\mathcal M$, and for
$\mathbf m\in\mathcal M\setminus\mathcal Z$ the integral solutions $y_A,y_B$ form a
fundamental pair.
\end{proposition}

\begin{proof}
By the standard local holomorphic dependence in the isomonodromic
Riemann--Hilbert correspondence
\cite{JimboMiwa1981,FokasItsKapaevNovokshenov2006}, the canonical columns of
the Lax pair depend holomorphically on $\mathbf m$.
Formula (\ref{eqP4RapidDecayCoefficientsArticle}) and the selected nonzero
normalization give the same dependence for the cycle coefficients. Estimate
(\ref{eqP4RapidDecayPhaseEstimateArticle}) is locally uniform in $\mathbf m$,
so the integrals $I_\nu,J_\nu$ and their determinant $\mathcal W$ are
holomorphic. The zero set of a holomorphic function that is not identically
zero is a proper analytic subset. Outside $\mathcal Z$, the nonzero constant Wronskian
gives a fundamental pair of solutions.
\end{proof}

Thus distinct asymptotic sectors exclude the identically vanishing mechanism
from Cauchy's theorem, while the proposition reduces generic nondegeneracy to
one nonzero value of $\mathcal W$. The computation below provides numerical
evidence for this hypothesis. A rigorous verification requires an exact
nonzero value or an error enclosure excluding zero, with asymptotic
initialization and infinite-tail errors included.

\subsection*{Numerical implementation of the integral representation}

For a direct numerical evaluation of
(\ref{eqP4LinearizedIntegralAnsatzArticle}), take
\[
\begin{gathered}
 x_0=0.23+0.07i,\qquad q(x_0)=1.17+0.19i,\qquad
 p(x_0)=0.41-0.16i,\\
\theta_0=0.37+0.11i,\qquad
 \theta_\infty=-0.23+0.07i.
\end{gathered}
\]
Here $i^2=-1$. The asymptotic tails were truncated at
$|\lambda|=R=4.8$, and the common finite vertex was $\lambda_*=1.45$.
Subdominant columns with the asymptotics
(\ref{eqP4SubdominantColumnsArticle}) were prescribed on the four rays
$\arg\lambda=0,\pi/2,\pi,-\pi/2$ and then continued along the paths
(\ref{eqP4ExplicitGammaLegsArticle}). At $x=x_0$, the normalized kernel
vector of the matrix of their symmetric squares was
\[
\begin{aligned}
c_0&=0.7173299462,&
c_1&=0.1234205969-0.5349056566i,\\
c_2&=0.1664299785+0.2966151084i,&
c_3&=0.2510505136-0.0733159090i.
\end{aligned}
\]
These coefficients were kept fixed for all subsequent values of $x$.

For $x=x_0+s$, $-0.048\le s\le0.048$, a uniform grid of 33 points with
step $\Delta s=0.003$ was used. The background quantities $q,h$ and the
Lax-pair columns were continued by the explicit eighth-order Runge--Kutta
method DOP853. Each integral
\[
 I_j(x)=\int_{\gamma_j}Q_{IV}(u_j(\lambda,x))\,d\lambda
\]
was evaluated by adaptive quadrature, separately on the radial segment and
the arc. The values of the linearized solution were defined by
\[
 y(x)=\sum_{j=0}^3c_jI_j(x),\qquad
 v(x)=\sqrt{q(x)}\,y(x),
\]
with the normalization $y(x_0)=1$.

Derivatives of the tabulated values were computed by fourth-order five-point
formulas:
\[
\begin{aligned}
 D_{\Delta s}y_i
 &=\frac{y_{i-2}-8y_{i-1}+8y_{i+1}-y_{i+2}}{12\Delta s},\\
 D_{\Delta s}^2y_i
 &=\frac{-y_{i+2}+16y_{i+1}-30y_i+16y_{i-1}-y_{i-2}}
 {12(\Delta s)^2}.
\end{aligned}
\]
For the normal-form equation, the relative residual was
\[
 \varepsilon_y=
 \frac{\max_i|D_{\Delta s}^2y_i-U_{IV}(x_i)y_i|}
 {\max_i\max\{|D_{\Delta s}^2y_i|,|U_{IV}(x_i)y_i|\}}
 =1.28\cdot10^{-7}.
\]
Denote the coefficient of $v$ in parentheses in
(\ref{eqLinP4ScalarArticle}) by
\[
 G_{IV}(x)=-\frac{(q')^2}{2q^2}+\frac92q^2+8xq
 +2(x^2-\alpha)-\frac{\beta}{q^2}.
\]
For $v_i=\sqrt{q_i}y_i$, we obtained
\[
 \varepsilon_v=
 \frac{\max_i|D_{\Delta s}^2v_i-(q'_i/q_i)D_{\Delta s}v_i
                   -G_{IV}(x_i)v_i|}
 {\max_i\max\{|D_{\Delta s}^2v_i|,
 |(q'_i/q_i)D_{\Delta s}v_i|,|G_{IV}(x_i)v_i|\}}
 =7.62\cdot10^{-8}.
\]
The maximum relative error in
(\ref{eqP4SymmetricSquareRelationArticle}) over the entire grid was
$1.84\cdot10^{-11}$, while the contribution of the truncated tails at
$|\lambda|=R$ was $1.15\cdot10^{-10}$. The result is shown in
Figure~\ref{figP4IntegralSolutionNumericsArticle}.

\begin{figure}[htbp]
\centering
\includegraphics[width=\textwidth]{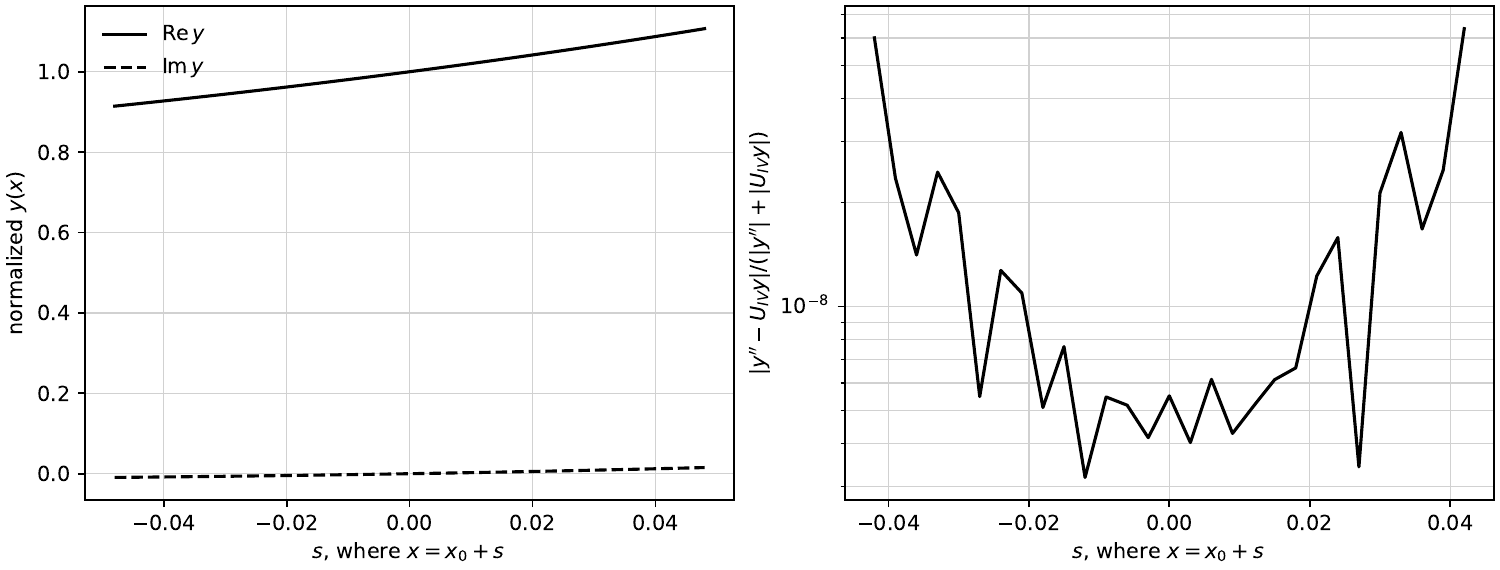}
\caption{Left: the solution $y(x)$ computed by direct quadrature of the
integral representation (\ref{eqP4LinearizedIntegralAnsatzArticle}) and
normalized by $y(x_0)=1$. Right: the pointwise relative residual of the
linearized equation in normal form.}
\label{figP4IntegralSolutionNumericsArticle}
\end{figure}
\FloatBarrier

\paragraph{Numerical verification of linear independence.}
For the second cycle we used the directions
$\arg\lambda=\pi/2,\pi,3\pi/2,2\pi$ and the common vertex
$\lambda_*^B=(1.45,\pi)$ on the universal cover. For the first cycle, the
directions were $-\pi/2,0,\pi/2,\pi$ and the common vertex was
$\lambda_*^A=(1.45,0)$. The normalized coefficients of the two linear
relations among the symmetric squares were
\[
\begin{aligned}
a_{-1}&=0.2510505139-0.0733159091i,
&a_0&=0.7173299452,\\
a_1&=0.1234205973-0.5349056572i,
&a_2&=0.1664299786+0.2966151089i,
\end{aligned}
\]
\[
\begin{aligned}
b_1&=0.3048535299+0.1143158736i,
&b_2&=-0.2533922373+0.0528314505i,\\
b_3&=0.9045628299,
&b_4&=-0.0925110192-0.0143253746i.
\end{aligned}
\]

The derivatives of the integral solutions were evaluated directly from
\[
 y_\nu'(x)=\int_{\Gamma^{(\nu)}}\partial_xQ_{IV}\,d\lambda,
 \qquad \nu=1,2,
\]
rather than by numerical differentiation of $y_\nu$. If
\[
 F=(h-\lambda)u_1^2-u_1u_2,
 \qquad u_x=B_{IV}u,
\]
then the total derivative of the integrand along a simultaneous solution of
the Lax pair has the explicit form
\[
 \partial_xQ_{IV}
 =\frac{\sqrt q}{2\lambda}\left\{
 \frac{q'}{2q}F+h'u_1^2+2(h-\lambda)u_1u_{1,x}
 -u_{1,x}u_2-u_1u_{2,x}\right\}.
\]
Each quadrature was normalized by $y_A(x_0)=y_B(x_0)=1$. The resulting
values were
\[
\begin{aligned}
 y_A'(x_0)&=1.9910828788+0.2520846155i,\\
 y_B'(x_0)&=2.0073638705+1.0159448467i,
\end{aligned}
\]
and hence
\[
\begin{aligned}
 W[y_A,y_B](x_0)&=0.0162809918+0.7638602311i,\\
 |W[y_A,y_B](x_0)|&=0.7640337188.
\end{aligned}
\]
The ratio of the Wronskian modulus to the maximum grid value of
$|y_Ay_B'|+|y_A'y_B|$ was $0.13049$, so the nonzero value is not caused by
loss of precision in subtracting nearly equal quantities. The maximum
relative variation of $W(x)/W(x_0)$ over the grid was
$8.15\cdot10^{-9}$. The relative residuals of $y''-U_{IV}y=0$ for $y_A$
and $y_B$ were $9.84\cdot10^{-8}$ and $5.53\cdot10^{-8}$, respectively.
Thus, two independent quadratures simultaneously verify the equation and the
condition $W[y_A,y_B]\ne0$. For the original linearization,
\[
 W[v_A,v_B](x_0)=q(x_0)W[y_A,y_B](x_0)
 =-0.1260846836+0.8968098589i\ne0.
\]

\begin{figure}[htbp]
\centering
\includegraphics[width=\textwidth]{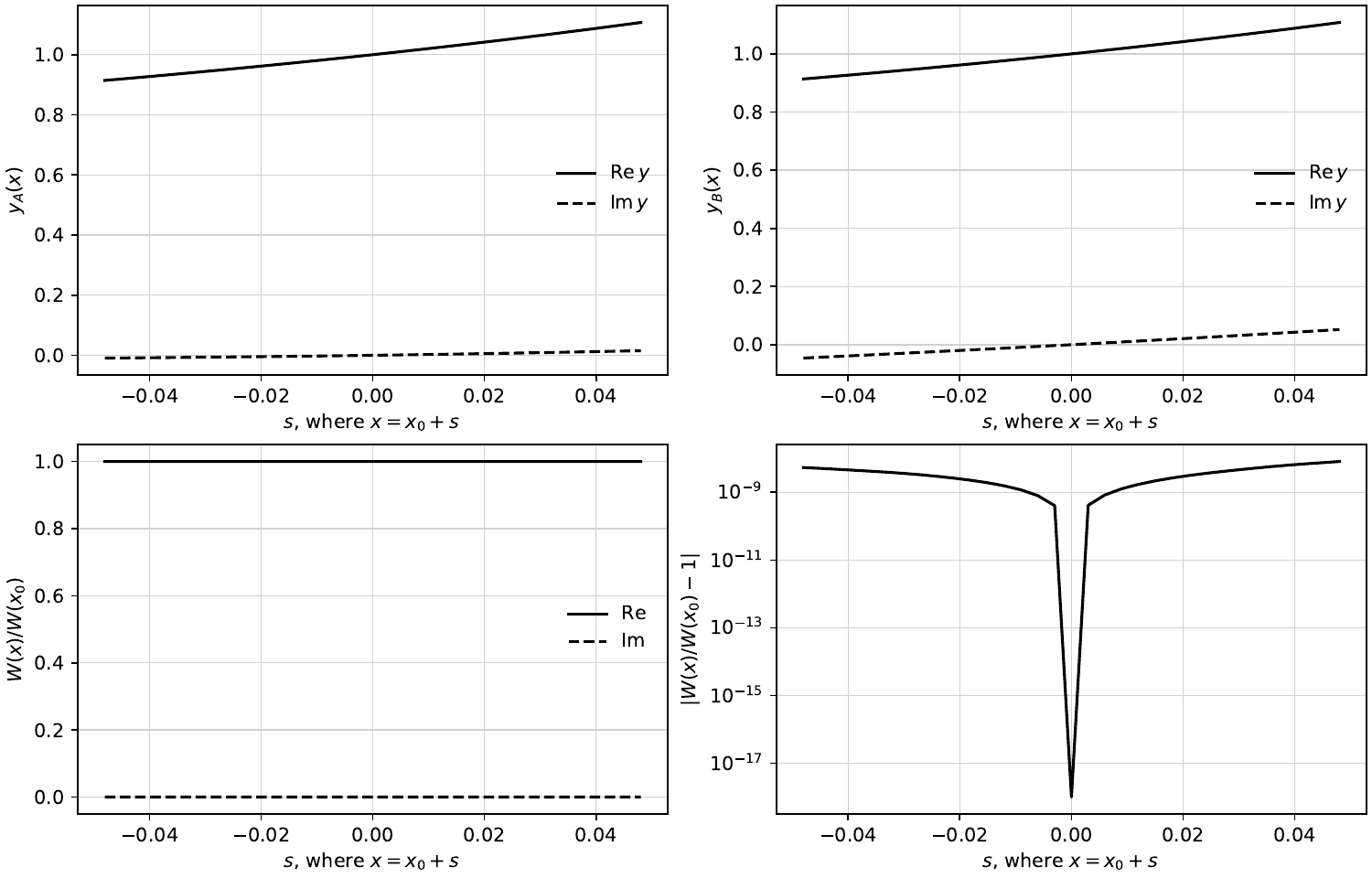}
\caption{Numerical verification of two contour solutions. The upper panels
show the normalized solutions $y_A$ and $y_B$. The lower left panel shows the
real and imaginary parts of $W(x)/W(x_0)$, and the lower right panel shows
the relative deviation of the Wronskian from its constant value. All
quantities were evaluated directly by contour quadrature, without integrating
the target linearized equation.}
\label{figP4TwoContourWronskianArticle}
\end{figure}
\FloatBarrier

\section{Variation of the monodromy data}

Consider an infinitesimal variation $q\mapsto q+\delta q$,
$p\mapsto p+\delta p$ of a solution of the fourth Painlev\'e equation. The
parameters $\theta_0,\theta_\infty$ are held fixed. In the notation of
(\ref{eqP4hArticle}),
\begin{equation}
    \delta h=2\delta p-\frac12\delta q.
    \label{eqP4DeltahArticle}
\end{equation}
The variation of the spectral part of the Lax pair is
\begin{equation}
    \delta A_{IV}
    =
    \begin{pmatrix}
        \delta a & \delta b\\
        \delta c & -\delta a
    \end{pmatrix},
    \label{eqP4DeltaAabcArticle}
\end{equation}
where
\begin{eqnarray}
    \delta a
    &=&
    \frac{h\,\delta q+q\,\delta h}{2\lambda},
    \label{eqP4DeltaaArticle}
    \\
    \delta b
    &=&
    -\frac{\delta q}{2\lambda},
    \label{eqP4DeltabArticle}
    \\
    \delta c
    &=&
    -(h\,\delta q+q\,\delta h)
    +
    \frac{
        \left(h^2+\frac{4\theta_0^2}{q^2}\right)\delta q
        +2qh\,\delta h
    }{2\lambda}.
    \label{eqP4DeltacArticle}
\end{eqnarray}
The variation of a fundamental solution satisfies
\[
    \frac{\partial \delta\Psi}{\partial\lambda}
    =
    A_{IV}\delta\Psi+\delta A_{IV}\Psi .
\]
Variation of constants gives the integral formula
\begin{equation}
    \delta\Psi
    =
    \Psi C_\delta
    +
    \Psi\int \Psi^{-1}\delta A_{IV}\Psi\,d\lambda .
    \label{eqP4DeltaPsiArticle}
\end{equation}
The matrix $C_\delta$ is independent of $\lambda$. Formula
(\ref{eqP4DeltaPsiArticle}) integrates a matrix-valued analytic one-form.
Fix a simply connected domain
\[
    \mathfrak D\subset\mathbb C\setminus\{0\},
\]
that does not cross the selected cuts of the Riemann--Hilbert problem, and
fix $\lambda_*\in\mathfrak D$. For $\lambda\in\mathfrak D$, set
\begin{equation}
    \mathcal I(\lambda;\lambda_*)
    =
    \int_{\lambda_*}^{\lambda}
    \Psi^{-1}(\mu)\delta A_{IV}(\mu)\Psi(\mu)\,d\mu ,
    \label{eqP4FiniteIntegralDefinitionArticle}
\end{equation}
where the integration path lies in $\mathfrak D$. Since the integrand is
analytic in $\mathfrak D$, the value of
(\ref{eqP4FiniteIntegralDefinitionArticle}) is invariant under deformations
with fixed endpoints. Formula (\ref{eqP4DeltaPsiArticle}) means that
\[
    \delta\Psi(\lambda)
    =
    \Psi(\lambda)C_\delta+\Psi(\lambda)\mathcal I(\lambda;\lambda_*),
\]
where $C_\delta=\Psi(\lambda_*)^{-1}\delta\Psi(\lambda_*)$ specifies
the normalization of the variation.

Under normalization at infinity, the entries of the matrix antiderivative
contain the exponentials $e^{\pm2\Theta_{IV}}$ and are integrated over the
different progressive tails defined in Section~5. For
\[
    Y_k(\lambda)=\Psi_k^{-1}(\lambda)\delta\Psi_k(\lambda)
\]
we therefore define, componentwise,
\begin{equation}
    J_{k,ij}(\lambda_*)
    :=
    \int_{\Gamma_k^{ij}(\lambda_*)}
    \left(\Psi_k^{-1}\delta A_{IV}\Psi_k\right)_{ij}\,d\lambda,
    \qquad
    J_k=(J_{k,ij})_{i,j=1}^2 .
    \label{eqP4InfinityIntegralDefinitionArticle}
\end{equation}
Each $\Gamma_k^{ij}(\lambda_*)$ is oriented from large $|\lambda|$ along
an admissible asymptotic ray toward the common finite point $\lambda_*$. The
ray is selected by the condition $Y_{k,ij}\to0$; then
$J_{k,ij}(\lambda_*)=Y_{k,ij}(\lambda_*)$.

Introduce the notation
\begin{equation}
    \mathcal V_k
    =
    \Psi_k^{-1}\delta A_{IV}\Psi_k
    =
    \bigl(\mathcal V_{k,ij}\bigr)_{i,j=1}^2 .
    \label{eqP4IkArticle}
\end{equation}
When $\det\Psi_k=1$, all entries of the matrix integrand are explicit:
\begin{equation}
\begin{aligned}
 \mathcal V_{k,11}
 &=\delta a(\Psi_{k,22}\Psi_{k,11}
             +\Psi_{k,12}\Psi_{k,21})
   +\delta b\,\Psi_{k,22}\Psi_{k,21}
   -\delta c\,\Psi_{k,12}\Psi_{k,11},\\
 \mathcal V_{k,12}
 &=\delta b\,\Psi_{k,22}^{\,2}
   -\delta c\,\Psi_{k,12}^{\,2}
   +2\delta a\,\Psi_{k,12}\Psi_{k,22},\\
 \mathcal V_{k,21}
 &=\delta c\,\Psi_{k,11}^{\,2}
   -\delta b\,\Psi_{k,21}^{\,2}
   -2\delta a\,\Psi_{k,11}\Psi_{k,21},\\
 \mathcal V_{k,22}&=-\mathcal V_{k,11}.
\end{aligned}
    \label{eqP4IkLimitDefinitionArticle}
\end{equation}

Varying $\Psi_{k+1}=\Psi_kS_k$ at $\lambda_*$ gives
\begin{equation}
    \delta S_k
    =
    S_kJ_{k+1}(\lambda_*)-J_k(\lambda_*)S_k .
    \label{eqP4DeltaSkMatrixArticle}
\end{equation}
For the constant Stokes matrices $\widetilde S_k$ of the joint solutions,
the corresponding variation is
\[
 \delta\widetilde S_k
 =\mathsf D^{-1}(\delta S_k)\mathsf D
   +[\widetilde S_k,\ell\sigma_3],\qquad
 \ell=\delta\log\mathfrak d=\frac12\int_{x_0}^x\delta q(t)\,dt.
\]
The reference point and the normalization $\mathfrak d(x_0)=1$ are held fixed.
Thus $\ell(x_0)=0$ and the spectral formulas evaluated at $x_0$ give
the variations of the constant monodromy coordinates directly.
The fundamental solution is normalized by $\det\Psi=1$. Since $S_1,S_3$
are lower triangular and $S_2,S_4$ are upper triangular,
\begin{equation}
\begin{aligned}
 \delta s_k={}&
 s_kJ_{k+1,11}+J_{k+1,21}
 -J_{k,21}-s_kJ_{k,22},
 \qquad k=1,3,
\end{aligned}
    \label{eqP4DeltaskLowerArticle}
\end{equation}
For the upper-triangular matrices, respectively,
\begin{equation}
\begin{aligned}
 \delta s_k={}&
 J_{k+1,12}+s_kJ_{k+1,22}
 -s_kJ_{k,11}-J_{k,12},
 \qquad k=2,4.
\end{aligned}
    \label{eqP4DeltaskUpperArticle}
\end{equation}
Here every $J_{j,ab}=J_{j,ab}(\lambda_*)$ is defined by the convergent
integrals (\ref{eqP4InfinityIntegralDefinitionArticle}) of the explicit
entries (\ref{eqP4IkLimitDefinitionArticle}).

Convergence of the integrals in (\ref{eqP4DeltaskLowerArticle}) and
(\ref{eqP4DeltaskUpperArticle}) follows from
\[
    \Theta_{IV}(\lambda,x)
    =
    \frac{\lambda^2}{2}+x\lambda-\theta_\infty\log\lambda,
\]
Away from the Stokes rays, the sign of $\operatorname{Re}\lambda^2$ therefore
determines which of the exponentials $e^{\pm2\Theta_{IV}}$ decays. The
canonical asymptotics show that $\mathcal V_{k,12}$ contains the factor
$e^{-2\Theta_{IV}}$; hence its tail $\Gamma_k^{12}$ is chosen in a component
where $\operatorname{Re}\Theta_{IV}\to+\infty$. Similarly,
$\mathcal V_{k,21}$ contains $e^{2\Theta_{IV}}$, and its tail
$\Gamma_k^{21}$ is chosen in a component where
$\operatorname{Re}\Theta_{IV}\to-\infty$. On closed subsectors, the
corresponding integrands are bounded by
$C|\lambda|^Ne^{-2c|\lambda|^2}$ with $C,c>0$ and $N\ge0$. Since the
exponents $\theta_0,\theta_\infty$ are fixed, substitution of the formal
series shows that the diagonal entries are $O(\lambda^{-2})$ and are also
absolutely integrable.

Formulas (\ref{eqP4DeltaskLowerArticle}) and
(\ref{eqP4DeltaskUpperArticle}) use the sectorial antiderivatives $J_k$ and
$J_{k+1}$ associated with different canonical solutions. A canonical column
that is subdominant at the first end generally has a dominant component at
the second, so the corresponding two-ended integral diverges. If two tails
lie in the same component, the analytic integral is zero by the
Cauchy-obstruction lemma.

For the connection matrix work at a nonzero comparison point $\lambda_*$
and assume $2\theta_0\notin\mathbb Z$. Choose the analytic eigenbasis
\[
 G_{00}=\begin{pmatrix}
 1&q/(4\theta_0)\\
 h-2\theta_0/q&qh/(4\theta_0)+1/2
 \end{pmatrix},\qquad \det G_{00}=1.
\]
It satisfies $A_{-1}G_{00}=G_{00}T$, where
$T=\theta_0\sigma_3$. Set
\[
 \Psi^{(0)}=G_0(\lambda)\lambda^T,\qquad
 G_0(\lambda)=\sum_{n=0}^\infty G_{0n}\lambda^n.
\]
The convergent Frobenius series and its variation are determined by
\[
 A_{-1}G_{00}=G_{00}T,\qquad
 (nI-A_{-1})G_{0n}+G_{0n}T
 =A_0G_{0,n-1}+\sigma_3G_{0,n-2}\quad(n\ge1),
\]
with $G_{0,-1}=0$. If
$\mathcal A_n(X)=(nI-A_{-1})X+XT$ and
$t_1=\theta_0$, $t_2=-\theta_0$, its inverse is explicitly
\[
 \mathcal A_n^{-1}(Y)=G_{00}
 \left(\frac{(G_{00}^{-1}Y)_{ij}}{n-t_i+t_j}\right)_{i,j=1}^2.
\]
Nonresonance makes every denominator nonzero for $n\ge1$.
At fixed exponents the initial variation is
\[
 \delta G_{00}=\begin{pmatrix}
 0&\delta q/(4\theta_0)\\
 \delta h+2\theta_0\delta q/q^2&
 (h\delta q+q\delta h)/(4\theta_0)
 \end{pmatrix},
\]
and successive coefficients satisfy
\[
 \mathcal A_n(\delta G_{0n})
 =\delta A_{-1}G_{0n}+\delta A_0G_{0,n-1}
  +A_0\delta G_{0,n-1}+\sigma_3\delta G_{0,n-2}.
\]
Here $\delta A_0$ and $\delta A_{-1}$ are respectively the constant
and residue coefficients of (\ref{eqP4DeltaAabcArticle}).
The Frobenius series converges normally in a sufficiently small disk,
locally uniformly in the initial data; its termwise variation therefore
converges there as well. At a nonzero point $\lambda_a$ in this disk set
\[
 Y_0(\lambda_a)=
 \lambda_a^{-T}G_0(\lambda_a)^{-1}\delta G_0(\lambda_a)\lambda_a^T.
\]
Continuation to $\lambda_*$ is given by the finite integral
\[
 Y_0(\lambda_*)=Y_0(\lambda_a)+
 \int_{\lambda_a}^{\lambda_*}
 (\Psi^{(0)})^{-1}\delta A_{IV}\Psi^{(0)}\,d\lambda.
\]
The path is on the selected lift and avoids zero. In general $Y_0$ does
not vanish at zero, because off-diagonal entries contain
$\lambda^{\pm2\theta_0}$.
Differentiating $\Psi_k=\Psi^{(0)}C_k$ gives
\begin{equation}
 \delta C_k=C_kJ_k(\lambda_*)-Y_0(\lambda_*)C_k.
 \label{eqP4DeltaConnectionArticle}
\end{equation}
This expression is independent of the finite comparison point.
For the local monodromy
\[
    M_0=C_1^{-1}e^{2\pi i\theta_0\sigma_3}C_1
\]
this gives
\begin{equation}
    \delta M_0
    =
    [M_0,\ C_1^{-1}\delta C_1].
    \label{eqP4DeltaM0Article}
\end{equation}
Similarly,
\[
    P=S_1S_2S_3S_4,\qquad M_\infty=F_\infty^{-1}P^{-1}
\]
gives
\begin{equation}
    \delta M_\infty
    =
    -M_\infty\left(\sum_{j=1}^{4}
    S_1\cdots S_{j-1}\,\delta S_j\,S_{j+1}\cdots S_4\right)P^{-1}.
    \label{eqP4DeltaMinftyArticle}
\end{equation}
If two canonically normalized solutions are related by $\Psi_+=\Psi_-S$,
set $J_\pm(\lambda_*)=\Psi_\pm^{-1}\delta\Psi_\pm(\lambda_*)$ with the
corresponding normalizations. The variation of the transition matrix is
\begin{equation}
    \delta S
    =
    SJ_+(\lambda_*)-J_-(\lambda_*)S .
    \label{eqP4DeltaMonodromyGeneralArticle}
\end{equation}
For a Stokes matrix, $J_+$ and $J_-$ are normalized in different Stokes
sectors; for a connection matrix, the local antiderivative has the finite initial
value determined by its normalized Frobenius series. The entries of the matrix integrand consist of squares and
pairwise products of entries of $\Psi$.

If the perturbed equation with a small parameter $\varepsilon$ is written as
\begin{equation}
    q''
    =
    \frac{(q')^2}{2q}
    +\frac32 q^3
    +4xq^2
    +2(x^2-\alpha)q
    +\frac{\beta}{q}
    +\varepsilon f(q,q',x),
    \label{eqP4PerturbedArticle}
\end{equation}
where $f(q,q',x)$ is a prescribed analytic function, then the first
correction satisfies an inhomogeneous linearized equation. After the
substitution $v=\sqrt q\,y$, its right-hand side is $f/\sqrt q$, and the
Green formula is constructed from solutions of
(\ref{eqLinP4SelfAdjointArticle}) represented by integrals of $Q_{IV}$.

\section{Asymptotic interpretation}

In the hierarchy of \cite{KiselevSuleimanov1999}, the Hamiltonian $H_{IV}$
describes the asymptotic dynamics of a solution, and $K_{IV}$ is its tangent
lift. The isomonodromic form of this construction is described by derivatives
of the solution with respect to the monodromy data and by integrals of
quadratic expressions from the Lax pair. The irregular singularity at
$\lambda=\infty$ determines the Stokes sectors, whereas the Fuchsian point
$\lambda=0$ determines transitions between sheets and the local monodromy.
These data enter both the integral solution
(\ref{eqP4LinearizedIntegralAnsatzArticle}) and the variational formulas
(\ref{eqP4DeltaskLowerArticle})--(\ref{eqP4DeltaM0Article}).

\section{Conclusion}

The scalar linearization has been reduced to the normal form
(\ref{eqLinP4SelfAdjointArticle}). The identity
(\ref{eqP4SquaredEigenfunctionIdentityArticle}) and a rapid-decay cycle give
a convergent integral solution with a vanishing boundary term; its numerical
implementation has relative residual $1.28\cdot10^{-7}$. A different cycle
class gives a second integral solution, and a nonzero Wronskian at one point
produces a fundamental pair of solutions. The Wronskian is holomorphic in the monodromy
data, so one nonzero value implies generic nondegeneracy; direct quadrature
provides numerical evidence for this condition for the selected complex data.
Canonical fundamental solutions of the associated linear system are defined
by Volterra equations, while the Stokes multipliers,
the connection matrix, and their variations are expressed by convergent
integrals.

\section*{Statements and declarations}

\paragraph{Conflict of interest.}
The author declares no conflict of interest.

\paragraph{Data availability.}
The article reports the numerical parameters and summary results.
The numerical verification code and machine-readable output are available
from the corresponding author upon reasonable request.

\section*{Acknowledgments}
The author received no specific funding for this work.

\end{document}